\documentclass[10pt,journal,compsoc]{IEEEtran}
\usepackage[pass]{geometry}
\usepackage{graphicx}
\usepackage{amsmath}
\usepackage{amsthm}
\usepackage{amssymb}
\usepackage{multirow}
\usepackage{color}
\usepackage{listings}
\usepackage{float}
\usepackage[nocompress]{cite}

\newtheorem{theorem}{Theorem}

\begin{document}
\title{EMOMA: Exact Match in One Memory Access}

\author{
Salvatore Pontarelli, Pedro Reviriego, Michael Mitzenmacher
\IEEEcompsocitemizethanks{
\IEEEcompsocthanksitem S. Pontarelli is with Consorzio Nazionale Interuniversitario per le Telecomunicazioni (CNIT), Via del Politecnico 1, 00133 Rome, Italy.\protect\\
E-mail: salvatore.pontarelli@uniroma2.it
\IEEEcompsocthanksitem P. Reviriego is with Universidad Antonio de Nebrija, C/ Pirineos, 55, E-28040 Madrid, Spain. \protect\\
E-mail: previrie@nebrija.es
\IEEEcompsocthanksitem M.  Mitzenmacher is with Harvard University, 33 Oxford Street, Cambridge, MA 02138, USA. \protect\\
E-mail: michaelm@eecs.harvard.edu
}
\thanks{Manuscript submitted 13 Set. 2017}}

\maketitle
\begin{abstract}
An important function in modern routers and switches is to perform a lookup
for a key.  Hash-based methods, and in particular cuckoo hash
tables, are popular for such lookup operations, but for large structures
stored in off-chip memory, such methods have the downside
that they may require more than one off-chip memory access to perform
the key lookup. Although the number of off-chip memory accesses 
can be reduced using on-chip approximate membership structures such as Bloom filters,
some lookups may still require more than one off-chip memory access. 
This can be problematic for some hardware implementations, as having 
only a single off-chip memory access enables a predictable processing 
of lookups and avoids the need to queue pending requests.

We provide a data structure for hash-based lookups based on cuckoo
hashing that uses only one off-chip memory access per lookup, by
utilizing an on-chip pre-filter to determine which of multiple
locations holds a key.  We make particular use of the flexibility to move
elements within a cuckoo hash table to ensure the pre-filter always gives
the correct response. While this requires a slightly more
complex insertion procedure and some additional memory accesses during insertions,
it is suitable for most packet processing applications where 
key lookups are much more frequent than insertions. An important feature of our approach is its simplicity. Our approach is based on simple logic that can be easily implemented in hardware, and hardware implementations would benefit most from the single  off-chip memory access per lookup.

\end{abstract}


\section{Introduction}
Packet classification is a key function in modern routers and switches
used for example for routing, security, and quality of service
\cite{1}.  In many of these applications, the packet is compared
against a set of rules or routes. The comparison can be an exact match,
as for example in Ethernet switching, or it can be a match with
wildcards, as in longest prefix match (LPM) or in a firewall rule.  The
exact match can be implemented using a Content Addressable Memory
(CAM) and the match with wildcards with a Ternary Content Addressable
Memory (TCAM) \cite{2,3}.  However, these memories are costly
in terms of circuit area and power and therefore alternative solutions
based on hashing techniques using standard memories are widely used
\cite{4}. In particular, for exact match, cuckoo hashing provides an
efficient solution with close to full memory utilization and a low and
bounded number of memory accesses for a match \cite{5}. For other
functions that use match with wildcards, schemes that use several
exact matches have also been proposed. For example, for LPM a binary
search on prefix lengths can be used where for each length an exact
match is done \cite{6}.  More general schemes have been proposed to
implement matches with wildcards that emulate TCAM functionality using
hash based techniques \cite{7}. In addition to reducing the circuit
complexity and power consumption, the use of hash based techniques
provides additional flexibility that is beneficial to support
programmability in software defined networks \cite{8}.

High speed routers and switches are expected to process packets with
low and predictable latency and to perform updates in the tables
without affecting the traffic. To achieve those goals, they commonly
use hardware in the form of Application Specific Integrated Circuits
(ASICs) or Field Programmable Gate Arrays (FPGAs) \cite{8,9}.
The logic in those circuits has to be simple to be able to process
packets at high speed. The time needed to process a packet has also to
be small and with a predictable worst case. For example, for
multiple-choice based hashing schemes such as cuckoo hashing, multiple
memory locations can be accessed in parallel so that the operation 
completes in one access cycle \cite{8}. This reduces
latency, and can simplify the hardware implementation by
minimizing queueing and conflicts. 

Both ASICs and FPGAs have internal memories that can be accessed with
low latency but that have a limited size. They can also be connected
to much larger external memories that have a much longer access time. 
Some tables used for packet processing are necessarily
large and need to be stored in the external memory, limiting the
speed of packet processing \cite{10}.  While parallelization may
again seem like an approach to hold operations to one memory access cycle,
for external memories parallelization can have a huge cost in terms
of hardware design complexity. Parallel access to external
memories would typically use different memory chips to perform parallel reads,
different buses to exchange addresses and data between the network
device and the external memory, and therefore a significant number of I/O pins are needed to
drive the address/data bus of multiple memory chips. Unfortunately, switch chips have a limited number of pins count and it seems that this limitation will be maintained over the next decade\cite{Binkert11}.
While the memory I/O interface
must work at high speed, parallelization is often unaffordable from the point
of view of the hardware design. When a single external  memory 
is used, the time needed to complete a lookup depends on the number of external 
memory accesses. This makes the hardware implementation
more complex if lookups are not always completed in one memory access cycle, and hence 
finding methods where lookups complete with a single memory access remains important 
in this setting to enable efficient implementations.  More generally, such schemes
may simplify or improve other systems that require lookup operations at large scale.

It is well known that in the context of multiple-choice hashing
schemes the number of memory accesses can be reduced by placing an
approximate membership data structure, such as a Bloom filter, as a
prefilter on the on-chip memory to guide where (at which choice) the
key can be found \cite{11}.  If we use a Bloom filter for each possible choice
of hash function to track which elements have been placed by each hash
function, a location in the external memory need only to be accessed when
the corresponding Bloom filter returns that the key can be found in
that location \cite{11}.  However, false positives from a Bloom filter
can still lead to requiring more than one off-chip memory access for a
non-trivial fraction of lookups, and in particular implies that more
than one lookup is required in the worst case.

We introduce an Exact Match in One Memory Access
(EMOMA) data structure, designed to allow a key
lookup with a single off-chip memory access.  We modify the prefilter
approach based on Bloom filters to tell us in which memory
location a key is currently placed by taking advantage of the cuckoo
hash table's ability to move elements if needed.  By moving elements, we can
avoid false positives in our Bloom filters, while maintaining the
simplicity of a Bloom filter based approach for hardware
implementation.  Our experimental results show that we can maintain
high memory loads with our off-chip cuckoo hash table.

The proposed EMOMA data structure is attractive for implementations that benefit 
from having a single off-chip memory access per lookup and applications that have a large ratio of lookups to insertions. Conversely, when more than one off-chip memory access can be tolerated for a small fraction of the lookups or when the number of insertions is comparable to that of lookups, other data structures will be more suitable.

Before continuing, we remark that our results 
are currently empirical; we do not have a theoretical proof regarding
for example the asymptotic performance of our data structure.  The
relationship and interactions between the Bloom filter prefilter and
the cuckoo hash table used in the EMOMA data structure are complex,
and we expect our design to lead to interesting future theoretical
work.

The rest of the paper is organized as follows. Section 2 covers the
background needed for the rest of the paper.  Most importantly, it
provides a brief overview of the relevant data structures (cuckoo hash
tables and counting block Bloom filters). Section 3 introduces our
Exact Match in One Memory Access (EMOMA) solution and discusses
several implementation options. EMOMA is evaluated in Section 4, where
we show that it can achieve high memory occupancy while requiring a
single off-chip memory access per lookup. Section 5 compares the proposed 
EMOMA solution with existing schemes. Section 6 presents the evaluation of 
the feasibility of a hardware implementation on an FPGA platform. Finally, 
Section 7 summarizes our conclusions and outlines some ideas for future work.



\section{Preliminaries}

This section provides background information on the memory
architecture of modern network devices and briefly describes two data
structures used in EMOMA: cuckoo hash tables and counting block Bloom
filters. Readers familiar with these topics can skip this section and
proceed directly to Section 3.

\subsection{Network Devices Memory Architecture}

The number of entries that network devices must store continues to
grow, while simultaneously the throughput and latency requirements
grow more demanding. Unfortunately, there is no universal
memory able to satisfy all performance requirements. On-chip SRAM on
the main network processing device has the highest throughout and
minimum latency, but the size of this memory is typically extremely small (few
MBs) compared with other technologies \cite{CheapSilicon13}. 
This is due in part to the larger size of the SRAM memory cells and to the fact that most of the
chip real estate on the main network processing device must be used for
other functions related to data transmission and switching. 
On-chip DRAMs (usually called embedded RAM or eDRAM) are currently used in
microprocessors to realize large memories such as L2/L3 caches
\cite{12}.  These memories can be larger (8x with respect to SRAM) but
have higher latencies. Off-chip memories such as DRAM have huge size
compared to on-chip memories (on the order of GB), but require power
consumption one order of magnitude greater that on-chip memory and
have latency higher than on-chip memories.  
For example, a Samsung 2Gb DRAM memory chip clocked at 1,866 MHz has worst case access time of 48 ns\footnote{Here we refer to the minimum time interval between successive active commands to the same bank of a DRAM. This time corresponds to the latency between two consecutive read accesses to different rows of the same bank of a DRAM.} \cite{samsung},\cite{Iyer03}.

Alternatives to standard
off-chip DRAM that reduce latency have been explicitly developed
for network devices. Some examples are the reduced latency DRAM
(RLDRAM) \cite{13} used in some Cisco Routers or the quad-data rate
(QDR) SRAM \cite{14} used in the 10G version of NetFPGA
\cite{9}. These memory types provide different compromises between
size, latency, and throughput, and can be used as second level memories
(hereinafter called external memories) for network devices.

Regardless of the type of memory used, it is important to minimize the average and worst case number of external memory accesses per lookup. As said in the introduction, having a single memory access per lookup simplifies the hardware implementation and reduces both latency and jitter.

Caching can be used, with the inner memory levels storing the most used entries \cite{10}. However, this approach does not improve the worst-case latency.  It also potentially creates packet reordering and packet jitter, and is effective only when the internal cache is big enough (or the traffic is concentrated enough) to catch a significant amount of traffic. 

Another option is to use the internal memory to store an approximate compressed information about the entries stored in the external memory to reduce the number of external memory accesses as done in EMOMA. This is the approach used for example in \cite{16}, where a counting Bloom filter identifies in which bucket a key is stored.  However, existing schemes do not guarantee that lookups are completed in one memory access or are not amenable to hardware implementation.
Section 5 discusses in more detail the benefits of EMOMA compared to other schemes based on using approximate compressed on-chip information.

\subsection{Cuckoo Hashing}
Cuckoo hash tables are efficient data structures commonly used to implement exact match \cite{5}. A cuckoo hash table uses a set of $d$ hash functions to access a table composed of buckets, each of which can store one or more entries. A given element $x$ is placed in one of the buckets $h_1(x)$, $h_2(x)$, \ldots , $h_d(x)$ in the table.
The structure supports the following operations:
\begin{itemize}
\item {\bf Search}: The buckets $h_i(x)$ are accessed and the entries stored there are compared with $x$;  if $x$ is found a match is returned.  
\item {\bf Insertion}: The element $x$ is inserted in one of the $d$ buckets.  If all the buckets are initially full, an element $y$ in one of the
buckets is displaced to make room for $x$ and recursively inserted.
\item {\bf Removal}: The element is searched for, and if found it is removed.
\end{itemize}
The above operations can be implemented in various ways. For example, typically on an insertion if the $d$ buckets for element $x$ are full a random bucket is selected and a random element $y$ from that bucket is moved.  
Another common implementation of cuckoo hashing is to split the cuckoo hash table 
into $d$ smaller subtables, with each hash function associated with (that is,
returning a value for) just one subtable.  The single-table and d-table alternatives provide the same asymptotic performance
in terms of memory utilization.  
When each subtable is placed on a different memory device this enables a parallel
search operation that can be completed in one memory access cycle
\cite{17}. However, as discussed in the introduction, this is not desirable for external memories, as
supporting several external memory interfaces requires increasing the
number of pins and memory controllers. 

It is possible that an element cannot be
placed successfully on an insertion in a cuckoo hash table. For example, when $d = 2$, if nine elements map
to the same pair of buckets and each bucket only has four entries,
there is no way to store all of the elements.  Theoretical results (as well
as empirical results) have shown this is a low probability
failure event as long as the load on the table remains sufficiently
small (see, e.g., \cite{18,19,20}).  This failure probability can be reduced significantly further
by using a small stash to store elements that would otherwise fail to be
placed \cite{21}; such a stash can also be used to hold elements
currently awaiting placement during the recursive insertion procedure,
allowing searches to continue while an insertion is taking place \cite{22}.

In cuckoo hashing, a search operation requires at most $d$ memory accesses. In the proposed EMOMA scheme, we use $d=2$. To achieve close to full occupancy with $d=2$, the table should support at least four entries per bucket. We use four entries per bucket in the rest of the paper.

\subsection{Counting Block Bloom Filters}

A Bloom filter is a data structure that provides approximate set membership checks using a table of bits \cite{23}.  We assume there are $m$ bits, initially all set to zero. To insert an element $x$, $k$ hash function values $h_1(x)$, \ldots $h_k(x)$ with range $[0,m-1]$ are computed and the bits with those positions in the table are set to 1. Conversely, to check if an element is present, those same positions are accessed and checked; when all of them are 1, the element is assumed to be in the set and a positive response is obtained, but if any position is 0, the element is known not to be in the set and a negative response is obtained.  The Bloom filter can produce false positive responses for elements that are not in the set, but false negative responses are not possible in a Bloom filter.

Counting Bloom filters use a counter in each position of the table instead of just a bit to enable the removal of elements from the set \cite{24}. The counters associated with the positions given by the $k$ hash functions are incremented during insertion and decremented during removal.  A match is obtained when all the counters are greater than zero. Generally, 4-bit counters are sufficient, although one can use more sophisticated methods to reduce the space for counters even further \cite{24}.  In the case of counting Bloom filters one option to minimize the use of on-chip memory is to use a normal Bloom filter (by converting all non-zero counts to the bit 1) on-chip while the associated counters are stored in external memory.  Search operations do not need to access the counters, but insertions and removals will require operations on the external memory; the benefit is that this reduces the amount of on-chip memory needed to implement the filter.

A traditional Bloom filter requires $k$ memory access to find a match. The number of accesses can be reduced by placing all the $k$ bits on the same memory word. This is done by dividing the table in blocks and using first a block selection hash function $h_0$ to select a block and then a set of $k$ hash functions $h_{1}(x)$, $h_{2}(x)$, \ldots $h_{k}(x)$ to select $k$ positions within that block \cite{25}. This variant of Bloom filter is known as a block Bloom filter.  When the size of the block is equal to or smaller than a memory word, a search can be completed in one memory access. Block Bloom filters can also be extended to support the removal of elements by using counters. In the proposed scheme, a counting block Bloom filter (CBBF) is used to select the hash function to use to access the external memory on a search operation, as we describe below.   

\section{Description of EMOMA}

EMOMA is a dictionary data structure that keeps key-value pairs $(x,v_x)$;  the structure can be queried to determine the value $v_x$ for a resident key $x$ (or it returns a null value if $x$ is not a stored key), and allows for the insertion and deletion of key-value pairs.  The structure is designed for a certain fixed size of keys that can be stored (with high probability), as explained further below.  We often refer to the key $x$ as an element.  When discussing issues such as inserting an element $x$, we often leave out discussion of the value, although it is implicitly stored with $x$.  

The EMOMA structure is built around a cuckoo hash table stored in external memory. In particular, two hash functions are used for the cuckoo hash table, and without any optimization two memory accesses could be required to search for an element. To reduce the number of memory accesses to one -- the main goal of EMOMA -- a counting block Bloom filter (CBBF) is used to determine the hash function that needs to be used to search for an element. Specifically, the CBBF keeps track of the set of elements that have been placed using the second hash function. 

On a positive response on the CBBF, we access the table using the second hash function, and otherwise, on a negative response, we access the table using the first hash function. As long as the CBBF is always correct, all searches require exactly one access to the external memory. A potential problem of this scheme is that a false positive on the CBBF would lead us to access the table using the second hash function when the element may have been inserted using the first hash function. This is avoided by ensuring that elements that would give a false positive on the CBBF are {\em always} placed according to the second hash function. That is, we avoid the possibility of a false positive leading us to perform a look up in the wrong location in memory by forcing the element to use the second hash function in case of a false positive, maintaining consistency at the cost of some flexibility.  In particular, such elements cannot be moved without violating the requirement that elements that yield a (false or true) positive on the CBBF must be placed with the second hash function. 

Two key design features make this possible. The first is that the CBBF uses the same hash function for the block selection as the first hash function for the cuckoo hash table. Because of this, entries that can create false positives on a given block in the CBBF can be easily identified, as we have their location in the cuckoo hash table. The second feature is that the cuckoo hash table provides us the flexibility to move entries so that the ones that would otherwise create false positives can be moved so that they are placed according to the second hash function. Although it may be possible to extend EMOMA for a cuckoo hash table that uses more than two hash functions, this is not considered in the rest of the paper. The main reason to do so is that in such configuration several CBBFs would be needed to identify the hash function to use for a search making the implementation more complex and less efficient.

The CBBF can be stored on-chip while the associated counters can be stored off-chip as they are not needed for search operations;  the counters will need to be modified for insertions or deletions of elements, however. The CBBF generally requires only one on-chip memory access as the block size is small and fits into a single memory word. The cuckoo hash table entries are stored off-chip. To achieve high utilization, we propose that the cuckoo hash table uses buckets that can contain (at least) four entries. As discussed in the previous section, two implementations are possible for the cuckoo hash table: a single table accessed with two hash functions or two independent subtables each accessed with a different hash function.  While in a standard cuckoo hash table both options are known to provide the same asymptotic performance in terms of memory occupancy, with our proposed data structure there are subtle reasons to be explained below that make the two alternatives different.  In the rest of the section the discussion focuses on the single-table approach but it can be easily extended for the double-table case.

\subsection{Structures}

The structures used in EMOMA for the single-table implementation are shown in Figure \ref{fig:diagram1T} and include:

\begin{enumerate}
\item A counting block Bloom filter (CBBF) that tracks all elements currently placed with the second hash function in the cuckoo table. The associated Bloom filter for the CBBF is stored on-chip and the counters are stored off-chip; we refer generally to the CBBF for both objects, where the meaning is clear by context. We denote the block selection function by $h_1(x)$ and the $k$ bit selection functions by $g_{1}(x)$, $g_{2}(x)$, \ldots $g_{k}(x)$.  The CBBF is preferably set up so that the block size is one memory word. 
\item A cuckoo hash table to store the elements and associated values; we assume four entries per bucket. This table is stored off-chip and accessed using two hash functions $h_1(x)$ and $h_2(x)$. The first hash function is the same as the one used for the block selection in the CBBF. This means that when inserting an element $y$ on the CBBF, the only other entries stored in the table that can produce a false positive in the CBBF are also in bucket $h_1(y)$. Therefore, they can be easily identified and moved out of the bucket $h_1(y)$ to avoid an erroneous response. 
\item A small stash used to store elements and their values that are pending insertion or that have failed insertion. The elements in the stash are checked for a match on every search operation. In what follows, think of the stash as a constant-sized structure.
\end{enumerate}

\begin{figure}[htb]
	\centering
	\includegraphics[width=0.45\textwidth]{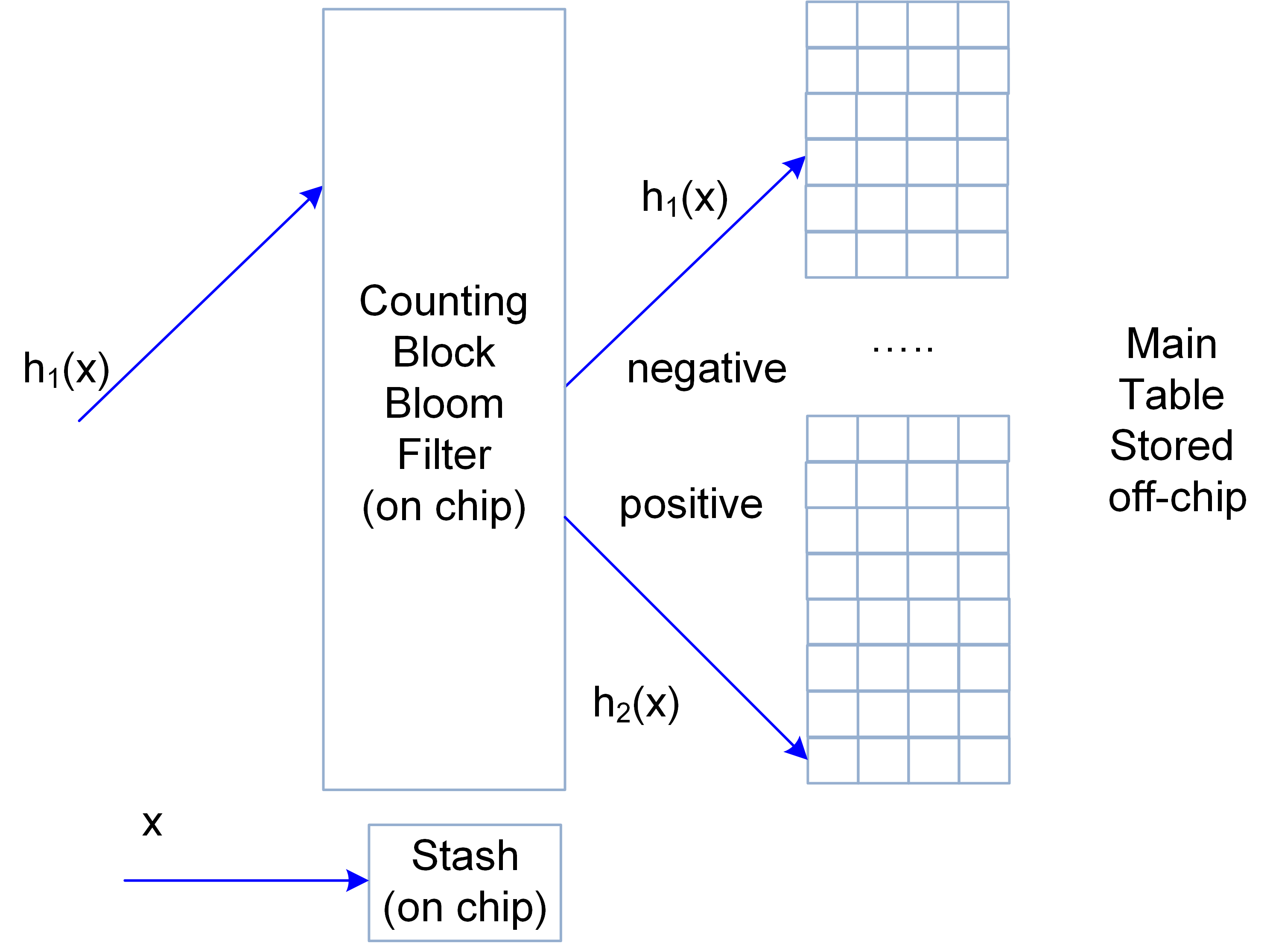} 
		\caption{Block diagram of the single-table implementation of the proposed EMOMA scheme.} 
	\label{fig:diagram1T}
\end{figure}

As mentioned before, an alternative is to place the elements on two independent subtables, one accessed with $h_1(x)$ and the other with $h_2(x)$. This double-table implementation is illustrated in Figure \ref{fig:diagram2T}. In this configuration, to have the same number of buckets, the size of each of the tables should be half that of the single table. Since the CBBF uses $h_1(x)$ as the block selection function, this in turn means that the CBBF has also half the number of blocks as in the single table case. Assuming that the same amount of on-chip memory is used for the CBBF in both configurations this means that the size of the block in the CBBF is double that of the single-table case. In the following, the discussion will focus on the single-table implementation but the procedures described can easily be modified for the double-table implementation.

\begin{figure}[htb]
	\centering
	\includegraphics[width=0.45\textwidth]{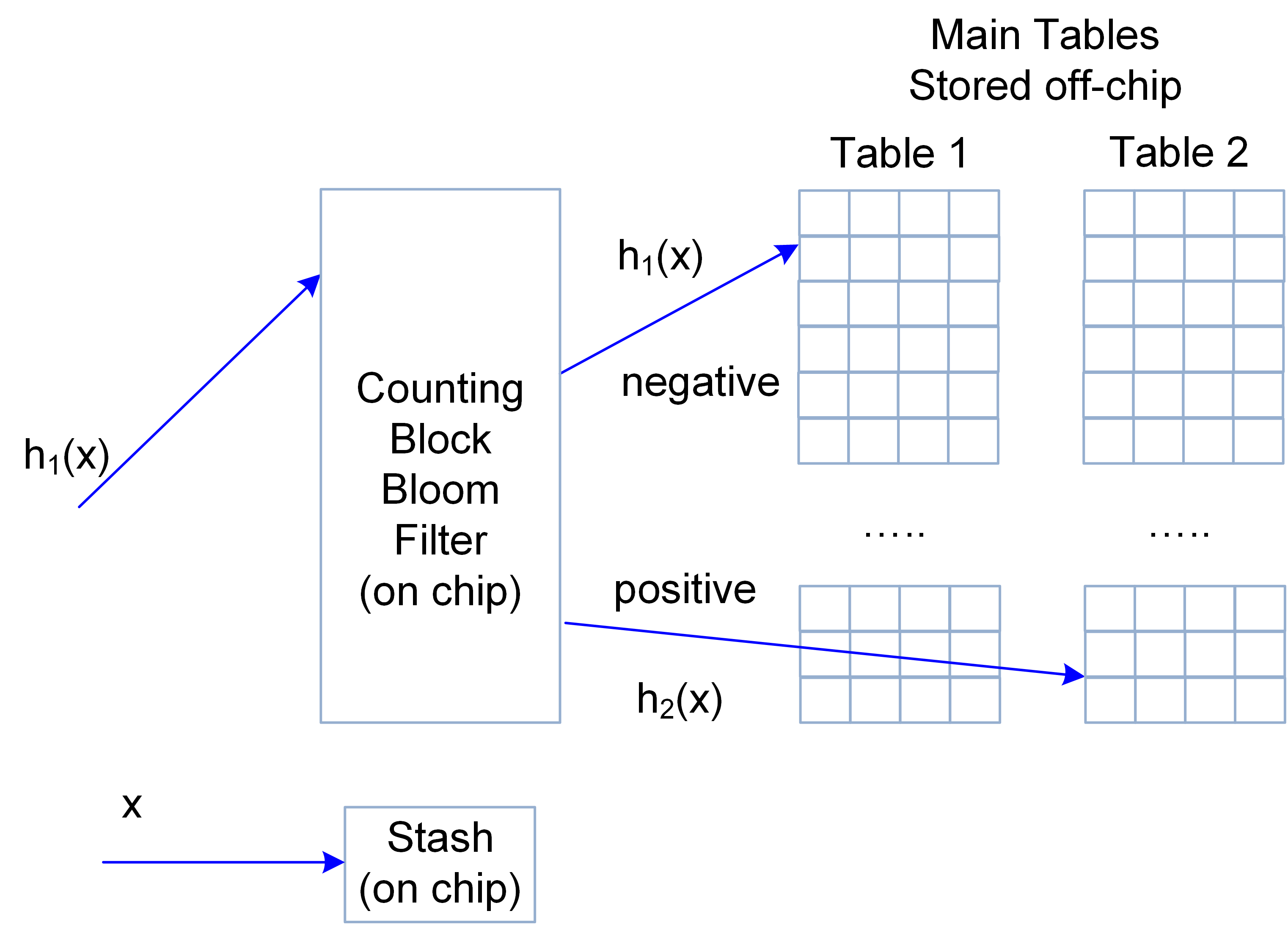} 
		\caption{Block diagram of the double-table implementation of the proposed EMOMA scheme.} 
	\label{fig:diagram2T}
\end{figure}

\begin{figure}[htb]
	\centering
	\includegraphics[width=0.27\textwidth]{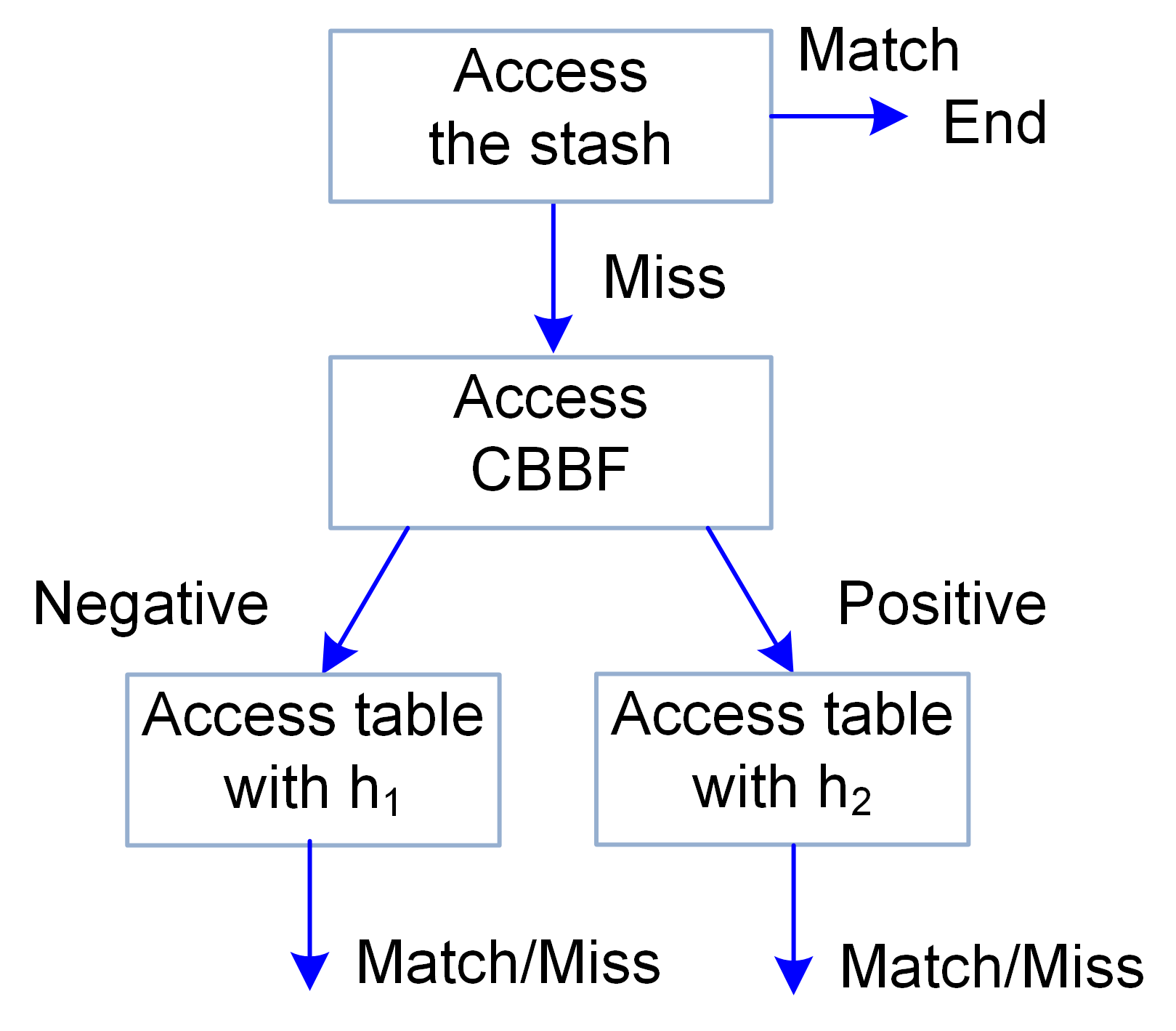} 
\caption{Search Operation.}
\label{fig:search}
\end{figure}

\subsection{Operations}

The process to search for an element $x$ is illustrated on Figure \ref{fig:search} and proceeds as follows:

\begin{enumerate}
\item The element is compared with the elements in the stash. On a match, the value $v_{x}$ associated with that entry is returned, ending the process.

\item Otherwise, the CBBF is checked by accessing position $h_1(x)$ and checking if the bits given by $g_{1}(x)$, $g_{2}(x)$, \ldots $g_{k}(x)$ are all set to one (a positive response) or not (a negative response).

\item On a negative response, we read bucket $h_1(x)$ in the hash table and $x$ is compared with the elements stored there. On a match, the value $v_{x}$ associated with that entry is returned, and otherwise we return a null value.

\item On a positive response, we read bucket $h_2(x)$ in the hash table and $x$ is compared with the elements stored there. On a match, the value $v_{x}$ associated with that entry is returned, and otherwise we return a null value.
\end{enumerate}

In all cases, at most one off-chip memory access is needed. 

Insertion is more complex. An EMOMA insertion must ensure that there are no false positives for elements inserted using $h_1(x)$, as any false positive would cause the search to use the second hash function when the element was inserted using the first hash function, yielding an incorrect response.  
Therefore we ensure that we place elements obtaining a positive response from the CBBF using $h_2(x)$. However, those elements can no longer be moved and therefore reduce the number of available moves in the cuckoo hash table, which are needed to maximize occupancy. In the following we refer to such elements as ``locked.'' As an example, assume now that a given block in the CBBF has already some bits set to one because previously some elements that map to that block have been inserted using $h_2(x)$. If we want to insert a new element $y$ that also maps to that block, we need to check the CBBF. If the response of this check is positive, this means that a search for $y$ would always use $h_2(y)$. Therefore, we have no choice but to insert $y$ using $h_2(y)$ and $y$ is ``locked'' in that bucket.  Locked elements  can only be moved if at some point elements are removed from the CBBF so that the locked element is no longer a false positive in the CBBF, thereby unlocking the element. Note that, to maintain proper counts in the CBBF for when elements are deleted, an element $y$ placed using the second hash function because it yields a false positive on the CBBF must still be added to the CBBF on insertion.

To minimize the number of elements that are locked, the number of elements inserted using $h_2(x)$ should be minimized as this reduces the number of ones on the CBBF and thus its false positive rate.  This fact seems to motivate using a single table accessed with two hash functions instead of the double-table implementation.  When two tables are used and we are close to full occupancy, at most approximately half the elements can be inserted using $h_1(x)$; with a single table, the number of elements inserted using $h_1(x)$ can be much larger than half. However, when two tables are used, the size of the block in the CBBF is larger making it more effective. Therefore, it is not clear which of the two options will perform better. In the evaluation section, results are presented for both options to provide insight into this question.

To present the insertion algorithm, we first describe the overall process and then discuss each of the steps in more detail. The process is illustrated in Figure \ref{fig:insertion} and starts when a new element $x$ arrives for insertion. The insertion algorithm will perform up to $t$ iterations, where in each iteration an element from the stash is attempted to be placed.  The steps in the algorithm are as follows:

\begin{enumerate}

\item Step 1: the new element $x$ is placed in the stash. This ensures that it will be found should any search operation for $x$ occur during the insertion.
\item Step 2: select a bucket to insert the new element $x$.
\item Step 3: select a cell in the bucket chosen in Step 2 to insert the new element $x$.
\item Step 4: insert element $x$ in the selected bucket and cell and update the relevant data structures if needed. Increase the number of iterations by one.
\item Step 5: Check if there are elements in the stash, and if the maximum number of iterations $t$ has not been reached.  If both conditions hold, select one of the elements uniformly at random and go to Step 2.  Otherwise, the insertion process ends.
\end{enumerate}

\begin{figure}[htb]
	\centering
	\includegraphics[width=0.45\textwidth]{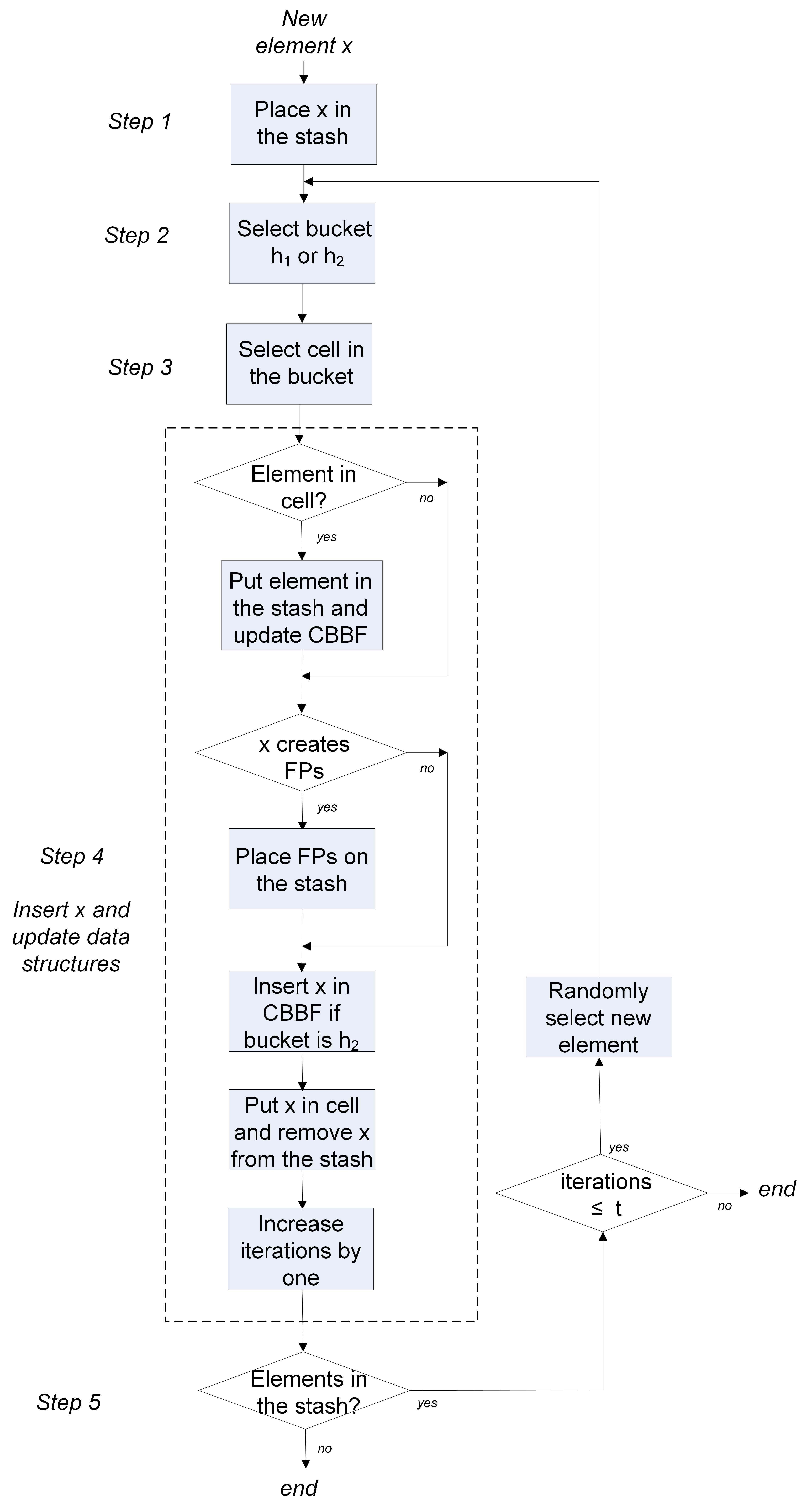} 
\caption{Insertion Operation.}
\label{fig:insertion}
\end{figure}

The first step to insert an element $x$ in EMOMA is to place it in the stash. This enables search operations to continue during the insertion as the new element will be found if a search is done. The same applies to elements that may be placed into the stash during the insertion as discussed in the following steps of the algorithm.

In the second step, we select one of the two buckets $h_1(x)$ or $h_2(x)$. The bucket selection depends on the following conditions:

\begin{enumerate}
\item   Are there empty cells in $h_1(x)$ and $h_2(x)$?
\item 	Is the element $x$ being inserted a false positive on the CBBF?
\item 	Does inserting $x$ in the CBBF create false positives for elements stored in bucket $h_1(x)$?
\end{enumerate}

Those conditions can be checked by reading buckets $h_1(x)$ and $h_2(x)$ and the CBBF block in address $h_1(x)$ and doing some simple calculations.  There are five possible cases for an insertion, as show in in Table \ref{t:insertion}.  (Note these cases are mutually exclusive and partition all possible cases.)  We describe these cases in turn.  

The first case occurs when $x$ itself is a false positive in the CBBF;  in that case, we must insert $x$ at $h_2(x)$ as on a search for $x$, the CBBF would return a positive and proceed to access the bucket $h_2(x)$.  This is illustrated on Figure \ref{fig:insert} (Case 1), where even if there is an empty cell in bucket $h_1(x)$ and there is no room in bucket $h_2(x)$, the new element $x$ must be inserted in $h_2(x)$ displacing one of the elements stored there. 

The second case occurs when the new element is not a false positive on the CBBF and there are empty cells in bucket $h_1(x)$. We then insert the new element on $h_1(x)$. This second case is illustrated on Figure \ref{fig:insert}  (Case 2).

The third case is when the new element $x$ is not a false positive on the CBBF, all the cells are occupied in bucket $h_1(x)$, there are empty cells on bucket $h_2(x)$, and inserting $x$ in the CBBF does not create false positives for other elements stored in bucket $h_1(x)$. Then $x$ is inserted in bucket $h_2(x)$ as shown in Figure \ref{fig:insert}  (Case 3).

The fourth case occurs when the new element $x$ is not a false positive on the CBBF, all the cells are occupied in bucket $h_1(x)$, and inserting $x$ on the CBBF creates false positives for other elements stored in bucket $h_1(x)$. The element is stored in bucket $h_1(x)$ to avoid the false positives even if there are empty cells in bucket $h_2(x)$. This is illustrated on Figure \ref{fig:insert}  (Case 4) where inserting $x$ in the CBBF would create a false positive for element $a$ that was also inserted in $h_1(x)$ (where $h_1(a)$ = $h_1(x)$).

Finally, the last case is when both buckets are full, the new element is not a false positive in the CBBF, and inserting it in the CBBF does not create other false positives. Then the bucket for the insertion is selected randomly, as both can be used. 

\begin{figure*}[htb]
\begin{tabular}{cccc}
	\includegraphics[width=0.22\textwidth]{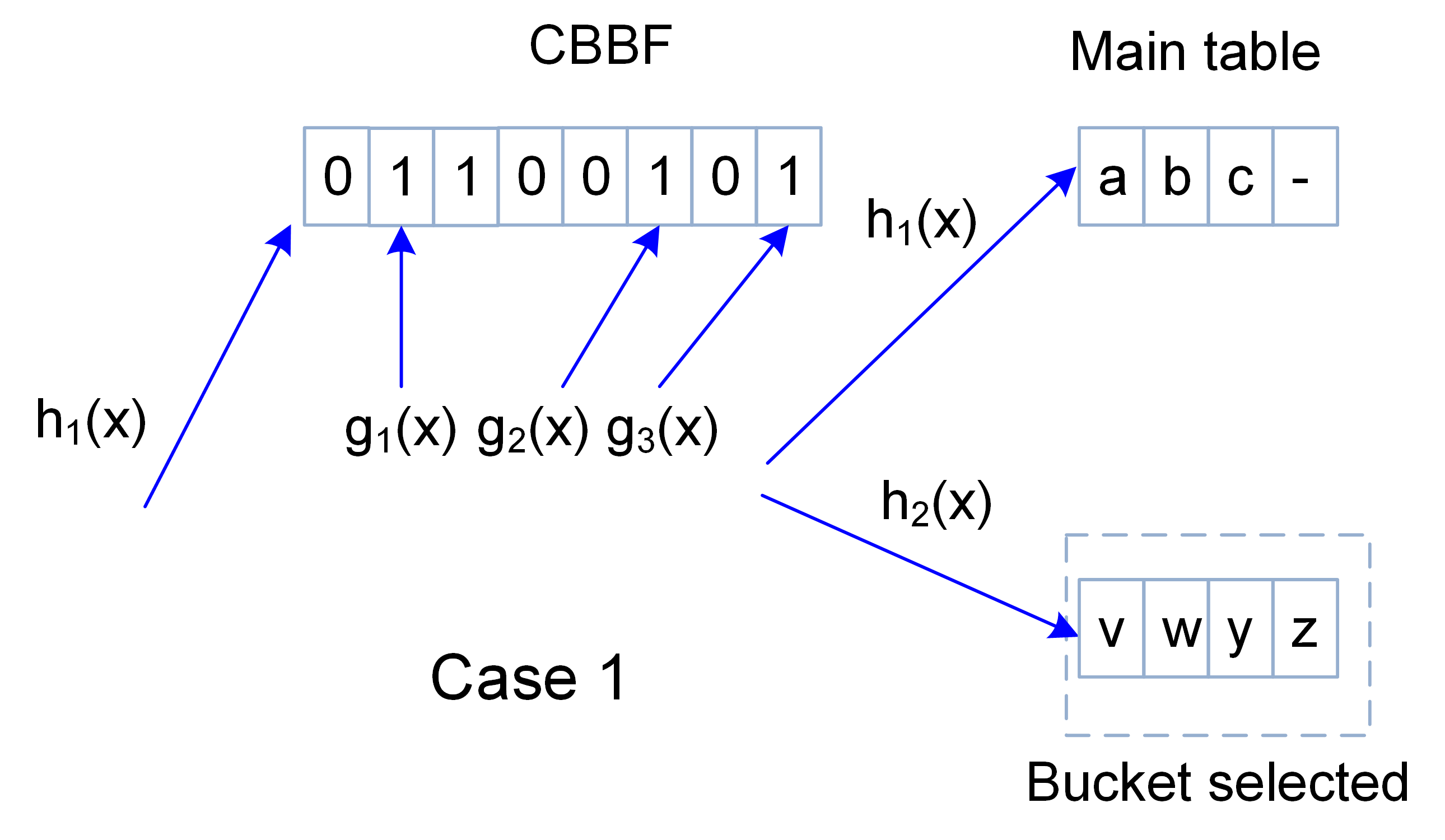} & 
	\includegraphics[width=0.22\textwidth]{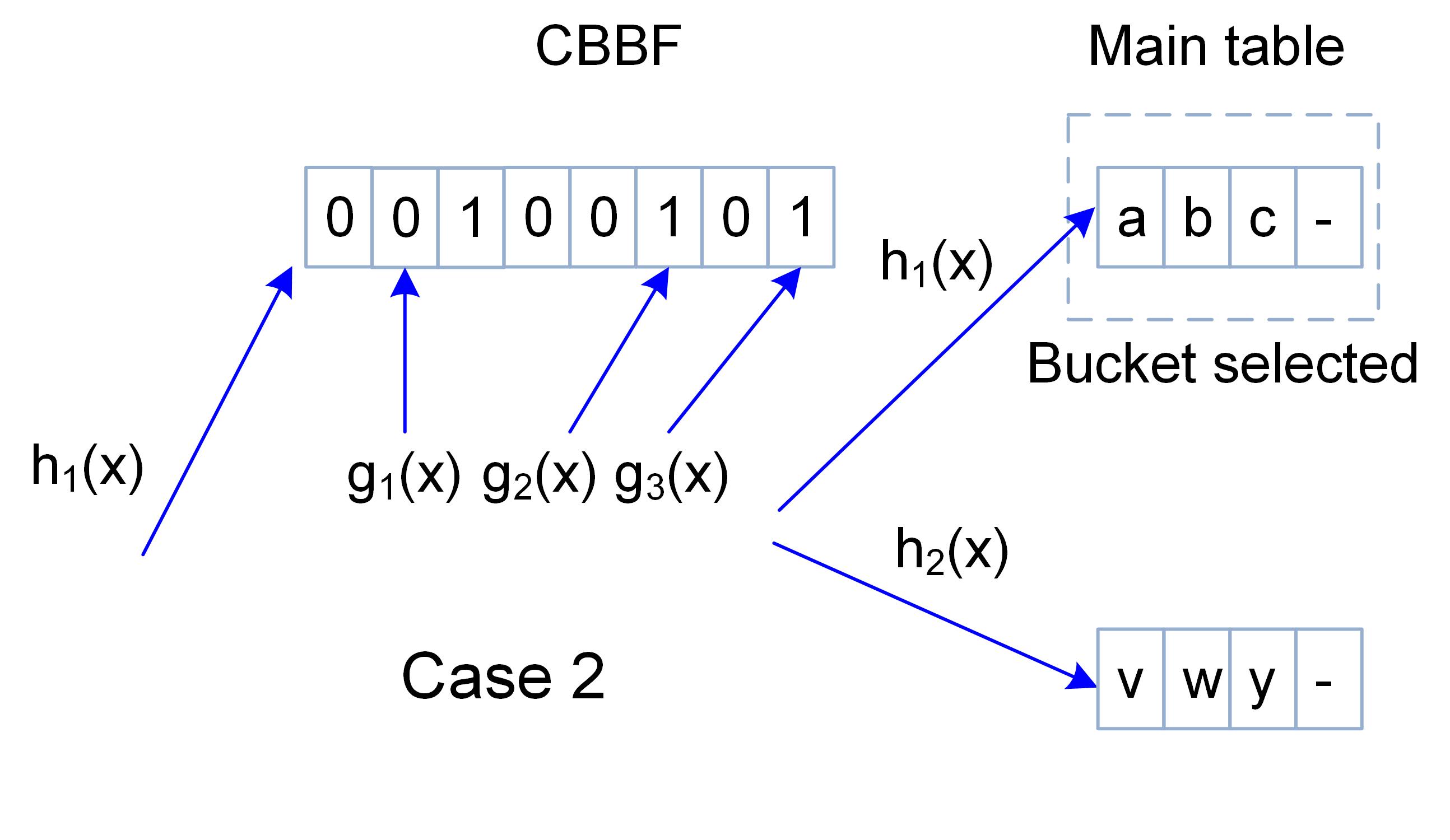} &
	\includegraphics[width=0.22\textwidth]{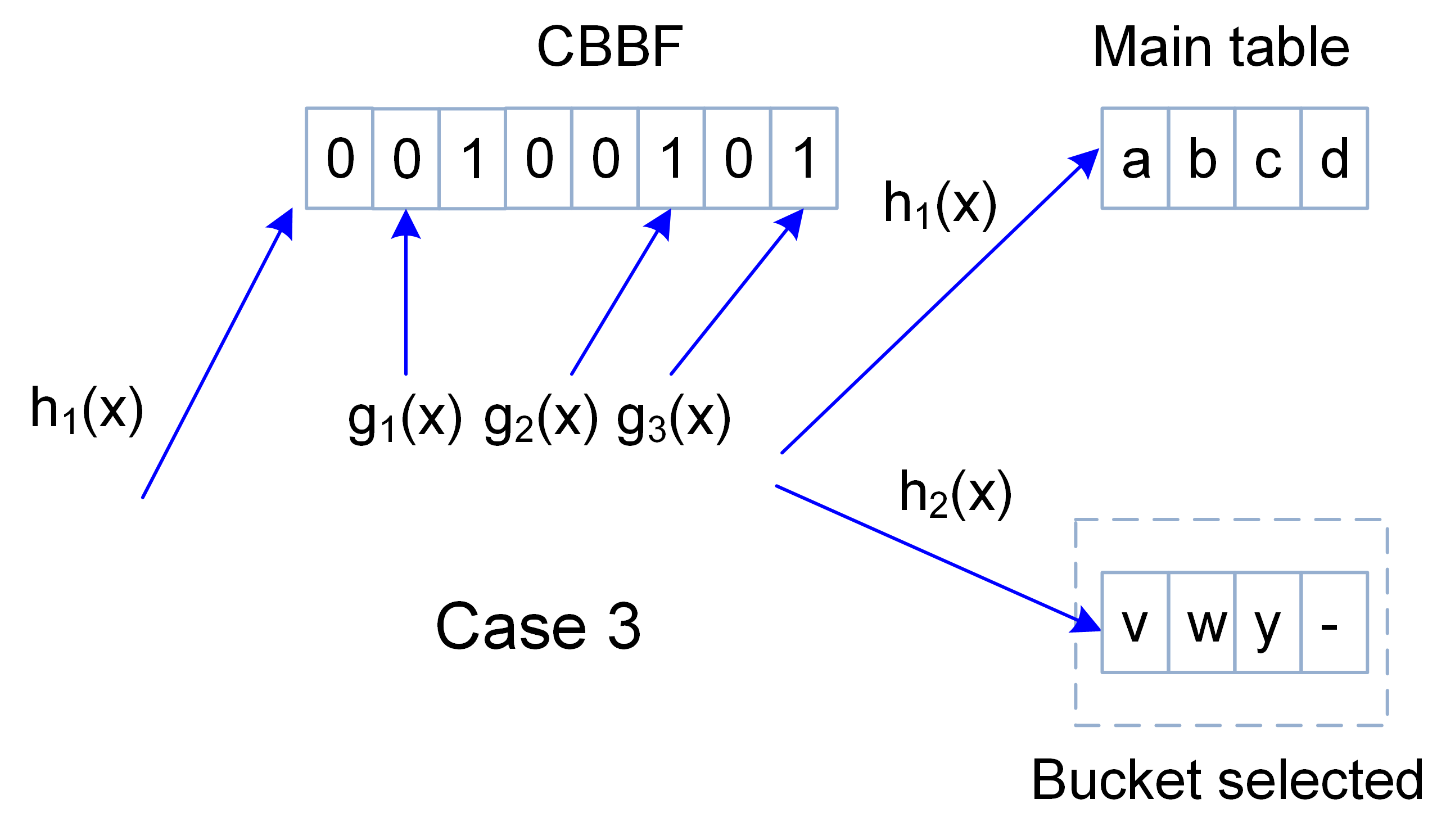} &
	\includegraphics[width=0.22\textwidth]{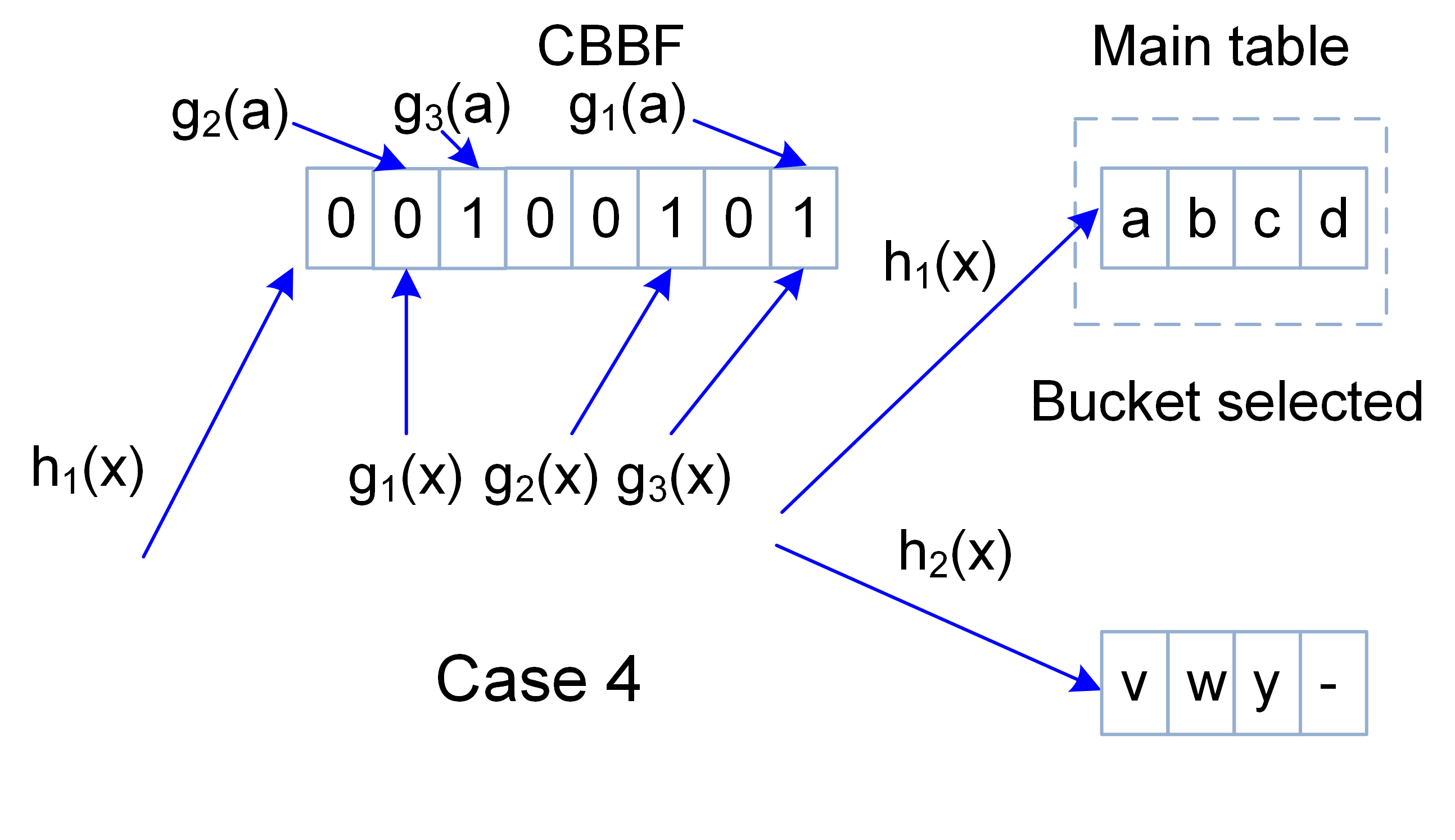} \\    
\end{tabular}
\caption{Examples of insertion of an element $x$ in EMOMA.}
\label{fig:insert}
\end{figure*}

\begin{table*}[t]
\centering
\caption{Selection of a bucket for insertion}
\begin{tabular}{|c|c|c|c|c|c|}
  \hline
Case & Empty cells& Empty cells & x is a false positive & Inserting x on the CBBF &  Bucket selected  \\
& in $h_1(x)$ &	 in $h_2(x)$ &	 on the CBBF &	creates false positives	&  for insertion  \\
\hline
Case 1	& Yes/No &	Yes/No &	Yes	& Yes/No &	$h_2(x)$ \\
\hline
Case 2 & Yes &	Yes/No &	No &	Yes/No &	$h_1(x)$   \\
\hline
Case 3	& No &	Yes  &	No &	No &	$h_2(x)$  \\
\hline
Case 4	& No &	Yes/No &	No &	Yes	 & $h_1(x)$  \\
\hline
Case 5	& No &	No & No & 	No &	Random selection \\
\hline
\end{tabular}
\label{t:insertion}
\end{table*}

The third step of the insertion algorithm selects a cell in the bucket chosen in the second step. This is done as follows:

\begin{enumerate}

\item If there are empty cells in the bucket, select one of them randomly. 
\item If all cells are occupied the selection is done among elements that are not locked as follows: with probability $P$ select randomly among elements that create the fewest locked elements when moved (elements inserted with $h_2$ will never create false positives). With probability $1-P$ randomly among all elements.
\end{enumerate}

It might seem that to reduce the elements that are locked during movements, we should set $P=1$.  Such a greedy approach of selecting an element to move that produces the fewest locked elements can limit flexibility, and can cause insertion failures that leave elements in the stash that could be placed. For example, if the element selected is $y$ and in bucket $h_2(y)$ there are four locked elements the insertion process will cycle until eventually halting and leaving additional elements in the stash as we will show in detail later, putting the data structure closer to failure. We corroborate this in the evaluation section.

Once the bucket and cell have been selected, the fourth step of the algorithm inserts element $x$ there. Before doing so, we need to check if there is an element $y$ stored in that cell. If so,  $y$ is placed in the stash and removed from the CBBF if it was inserted using $h_2(y)$. This may unlock elements that are no longer false positives on the CBBF due to the removal of $y$ from the CBBF;  such elements remain in the second table, however. We also need to check if $x$ is inserted into $h_2(x)$ that, as a result of inserting $x$, elements in bucket $h_1(x)$ need to be moved (or locked) because they will be false positives on the CBBF once $x$ is inserted. If so they are also placed in the stash. Then $x$ is inserted in the CBBF if the selected bucket is $h_2(x)$, and finally $x$ is inserted on the selected cell and removed from the stash. The number of iterations is increased by one before proceeding to the next step.

In the fifth and last step of the insertion algorithm, we check if there are elements in the stash (either because they are placed there while inserting $x$, or if they have been left there from previous insertion processes). If there are any elements in the stash, and the maximum number of insertion iterations $t$ has not been been performed, then we select randomly one of the elements in the stash and return to the second step.  Otherwise, the insertion process ends.
The number of iterations affects the time for an insertion process as well as the size of the stash that is needed.  Generally, the more iterations, the longer an insertion can take, but the smaller a stash required.  We explore this trade-off in our experiments.  Elements may be left in the stash at the end of the insertion process. If the stash ever fails to have enough room for elements that have not been placed, the data structure fails. The goal is that this type of failure should be a low probability event.

As with most hashing-based lookup data structures, insertion is more complex than search. Fortunately, in most networking applications, insertions are much less frequent than searches. For example, in a router, the peak rate of BGP updates is in the order of thousands per second, while the average rate is a few insertions per second \cite{bgp1, bgp2}. On the other hand, a router can perform several million packet lookups in a second. Similar or smaller update rates occur in other network applications such as MAC learning or reconfiguration of OpenFlow tables. 



Removing an element starts with a search and if the element is found it is removed from the table. If the element's location was given by the second hash function, the element is also removed from the CBBF by decreasing the counters associated with bits $g_{1}(x)$, $g_{2}(x)$, \ldots $g_{k}(x)$ in position $h_1(x)$. If any counter reaches zero, the corresponding bit in the bit (Bloom filter) representation of the CBBF is cleared. The removal of elements from the CBBF may unlock elements previously locked on their second bucket if they are no longer false positives on the CBBF;  however, such unlocked elements are not readily detected, and will not be moved to the bucket given by their first hash function until possibly some later operation.  A potential optimization would be to periodically scrub the table looking for elements $y$ stored in position $h_2(y)$ and moving them to position  $h_1(y)$ if they are not false positives on the CBBF and there are empty cells on bucket $h_1(y)$.  We do not explore this potential optimization further here.

As mentioned before, a key feature in EMOMA is that the first hash function used to access the hash table is also used as the block selection function on the CBBF. Therefore, when we insert an element in the table using the second hash function, the elements that can result in a false positive in the Bloom filter as a result can be easily identified;  they are in the bucket indexed by $h_1(x)$ that were inserted there using their own first hash function. To review, the main differences of EMOMA versus a standard cuckoo hash with two tables are: 

\begin{itemize}
\item  Elements that are false positives in the CBBF are ``locked'' and can only be inserted in the cuckoo hash table using the second hash function. This reduces the number of options to perform movements in the table.
\item Insertions in the cuckoo hash table using the second hash function can create new false positives for the elements in bucket $h_1(x)$ that require additional movements. Those elements have to be placed in the stash and re-inserted into the second table. This means that, in contrast to standard cuckoo hashing, the stash occupancy can grow during an insertion. Therefore, the stash needs to be dimensioned to accommodate those elements in addition to the elements that have been unable to terminate insertion. 
\end{itemize}

The effect of these differences depends mainly on the false positive rate of the CBBF. That is why the insertion algorithm aims to minimize the number of locked elements. In the next section, we show that even when the number of bits per entry used in the CBBF is small, EMOMA can achieve memory occupancies of over 95\% with 2 bucket choices per element and 4 entries per bucket. A standard cuckoo hash table can achieve memory occupancies of around 97\% with 2 choices per element and 4 entries per bucket.  The required stash size and number of movements needed for the insertions also increase compared to a standard cuckoo hash but remain reasonable. 
Therefore, the restrictions created by EMOMA for movements in the cuckoo hash table have only a minor effect in practical scenarios.  Theoretically analyzing the effect of the CBBF on achievable load thresholds for cuckoo hash tables remains a tantalizing open problem.  

We formalize our discussion with this theorem.
\begin{theorem}
When all elements have been placed successfully or lie in the stash, the EMOMA data structure completes search operations with one external memory access.
\end{theorem}
\begin{proof}
As only one bucket is read on a search, we argue that if an element $x$ is stored in the table, the search operation will always find it. If $x$ is stored in bucket $h_1(x)$, then EMOMA will fail to find it only if the CBBF returns a positive on $x$. This is not possible as elements that are positive on the CBBF are always inserted into $h_2(x)$, as can be seen by examining all the cases in the case analysis. Similarly, if an element $x$ is stored in bucket $h_2(x)$, then a search operation for $x$ will fail only if $x$ is not a positive on the CBBF. Again, this is not possible as elements inserted using $h_2(x)$ are added to the CBBF. These properties hold even when (other) elements are removed. When another element $y$ is removed, it is also removed from the CBBF if it was stored on its second bucket. If $x$ was a negative in the CBBF, it will remain so after the removal.  If $x$ was a positive in the CBBF, even if it was originally a false positive it was added into the CBBF to make it a true positive, and thus the CBBF result for $x$ does not depend on whether other elements are stored or not on the CBBF.
\end{proof}

\section{Evaluation of EMOMA}

We have implemented the EMOMA scheme in C++ to test how its behavior depends on the various design parameters and to determine how efficiently it uses memory in practical settings.  Since all search operations are completed in one memory access, the main performance metrics for EMOMA are the memory occupancy that can be achieved before the data structure fails (by overflowing the stash on an insertion) and the average insertion time of an element. The parameters that we analyzed are:
\begin{itemize}
\item The parameter $P$ that determines the probability of selecting an element to move randomly, as described previously.
\item 	The number of bit selection hash functions $k$ used in the CBBF.
\item 	The number of tables used in the cuckoo hash table (single-table or double-table implementations).
\item 	The number of on-chip memory bits per element (bpe) in the table, which determines the relative size of the CBBF versus the off-chip tables.
\item 	The maximum number of iterations $t$ allowed during an insertion before stopping and leaving the elements in the stash. These insertions are referred in the following as non-terminating insertions.
\item 	The size of the stash needed to avoid stash overflow.
\end{itemize}

We first present simulations showing the behavior of the stash with respect to the $k$ and $P$ parameters for three table sizes (32K, 1M and 8M, where we conventionally use 1K for $2^{10}$ elements and 1M for $2^{20}$ elements.). We then present simulations to evaluate the stash occupancy when the EMOMA structure works at a high load (95\%) under dynamic conditions (repeated insertion and removal of elements). We also consider the average insertion time of the EMOMA structure. Finally, we estimate how the size of the stash varies with table size and present an estimation of the failure probability due to stash overflow. In order to better understand the impact of the EMOMA scheme on the average insertion time and the stash occupancy we compared the obtained results with corresponding results using a standard cuckoo hash table.

\subsection{Parameter selection}

Our first goal is to determine generally suitable values for the number of hash functions $k$ in the CBBF and the probability $P$ of selecting an element to move randomly;  we then fix these values for the remainder of our experiments.  For this evaluation, we generously overprovision a stash size of 64 elements, although in many configurations EMOMA can function with a smaller stash. The maximum stash occupancy during each test is logged and can be used for relative comparisons. A larger stash occupancy means that those parameter settings are more likely to eventually lead to a failure due to stash overflow.  

We first present two experiments to illustrate the influence of $P$ and $k$ on performance. In the first experiment, two small tables that can hold 32K elements each were used, $k$ was set to four, and four bits per element were used for the CBBF while $P$ varied from 0 to 1. The maximum number of iterations for each insertion $t$ is set to 100.

For each configuration, the maximum stash occupancy was logged and the simulation inserted elements until a 95\% memory use was reached. The simulation was repeated 1000 times. Figure~\ref{fig:graph1} shows the average across all the runs of the maximum stash occupancy observed. The value of $P$ that provides the best result is close to 1, but too large a value of $P$ yields a larger stash occupancy. This confirms the discussion in the previous section; in most cases it is beneficial to move elements that create the least number of false positives but a purely greedy strategy can lead to unfortunate behaviors. From these results it appears that a value of $P$ in the range 0.95 to 0.99 provides the best results.  

\begin{figure}[t]
	\centering
	\includegraphics[width=0.48\textwidth]{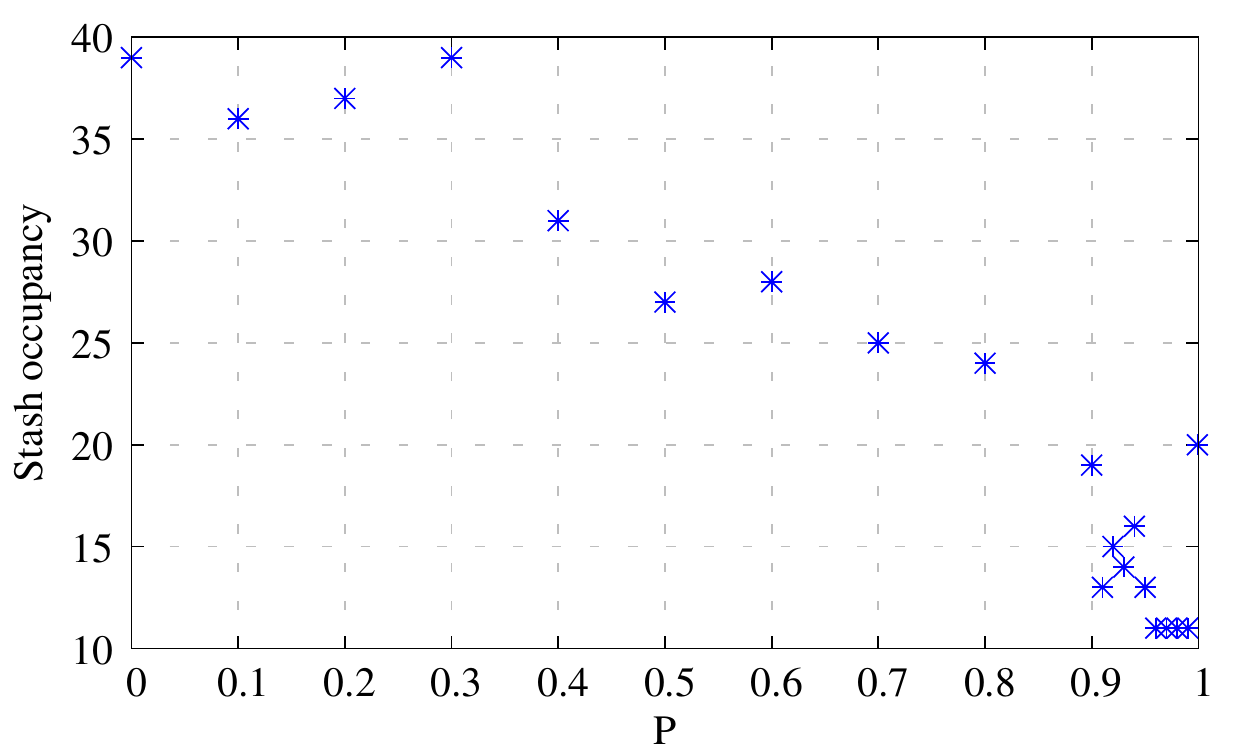} 
\caption{Average of the maximum stash occupancy over 1000 runs for different values of $P$ at 95\% memory occupancy, single-table, $k=4$, bpe=4 and $t=100$.}
\label{fig:graph1}
\end{figure}

In the second experiment, we set $P=0.99$ and we varied $k$ from 1 to 8. The results for the single-table configuration are shown in Figure~\ref{fig:graph2}. In this case, the best values were $k= 3,4$ when the double-table implementation is used and $k=3$ when a single table is used. However, the variation as $k$ increases up to 8 is small. (Using $k=1$ provided poor performance.)  Based on the results of these two smaller experiments, the values $P=0.99$ and $k=3$ for the single-table variant and $k=4$ for the double-table variant are used for the rest of the simulations. 

\begin{figure}[t]
	\centering
	\includegraphics[width=0.48\textwidth]{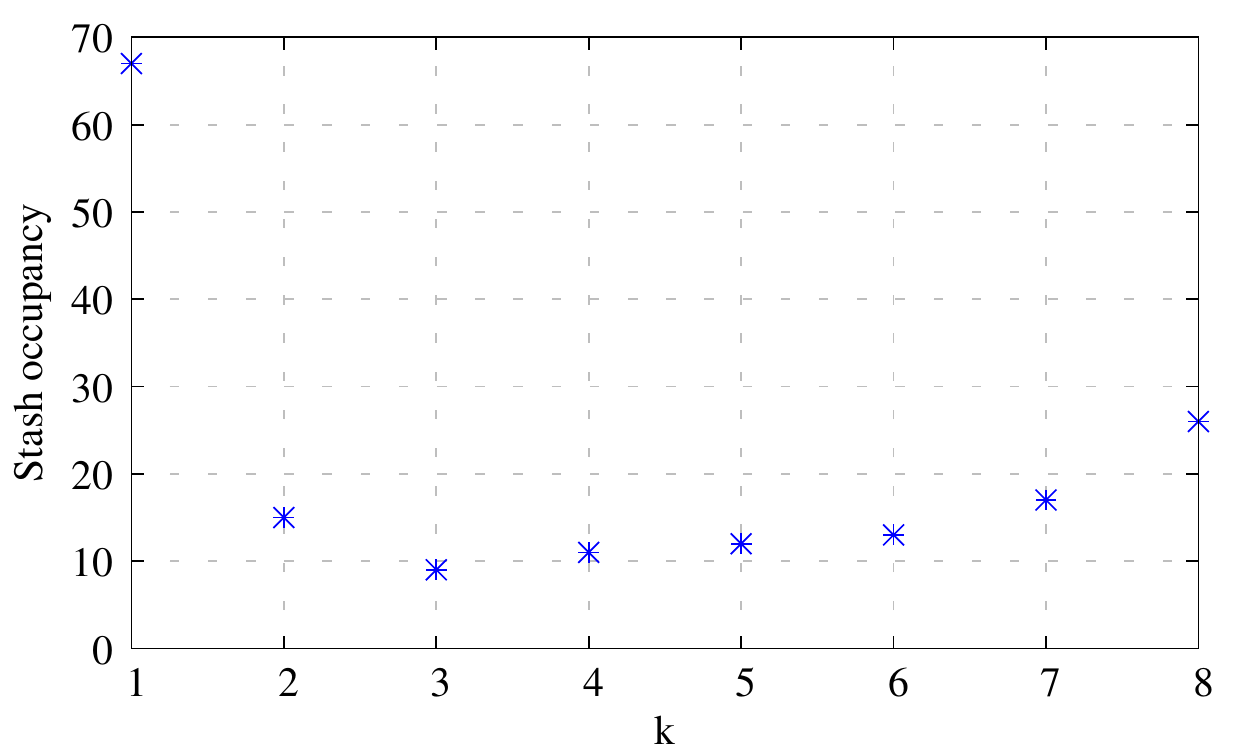} 
\caption{Average of the maximum stash occupancy over 1000 runs for different values of $k$ at 95\% memory occupancy, single-table, $P$=0.99, bpe=4 and $t$=100.}
\label{fig:graph2}
\end{figure}

Given these choices of $P$ and $k$, we aim to show that EMOMA can reliably achieve 95\% occupancy in the cuckoo hash table using four on-chip memory bits per element for the CBBF. We test this for cuckoo hash tables of sizes 32K, 1M, and 8M elements, with both single-table and double-table implementations. In particular, we track the maximum occupancy of the stash during the insertion procedure in which the table is filled up to 95\% of table size. The distribution of the stash occupancies over 1000 runs are shown in Figure ~\ref{fig:graph3}. 


In all cases, the maximum stash size observed is fairly small. The maximum values for the single-table option were 9, 14, and 16 for table sizes 32K, 1M, and 8M respectively. For the double-table option, these maxima were 9, 18, and 33. These results suggest that the single-table option is better, especially for large table sizes.

We also looked at the percentage of elements stored using $h_1(x)$ and $h_2(x)$. In the single-table implementation, the percentages were 59\% and 41\% respectively, while in the double-table implementation, the percentages were 52\% and 48\%. These results show how the use of a single table enables placing more elements using the first hash function, thereby reducing the false positive rate in the CBBF and thus the number of elements locked. This confirms our previous intuition. In fact, the use of a single-table has another subtle benefit: when inserting an element $x$ using $h_2(x)$, of the elements in bucket $h_1(x)$, only those inserted there with $h_1$ can cause a false positive. With two tables, all the elements in the first table in bucket $h_1(x)$ can cause a false positive. Therefore on average the single-table implementation has fewer candidates to create false positives than the double-table implementation for each insertion using $h_2$. These factors tend to make the single-table option better, as will be further seen in our remaining simulation results.  We therefore expect that the single-table variant will be used in practical implementations.


\begin{figure*}[t]
	\centering
	\includegraphics[width=0.3\textwidth]{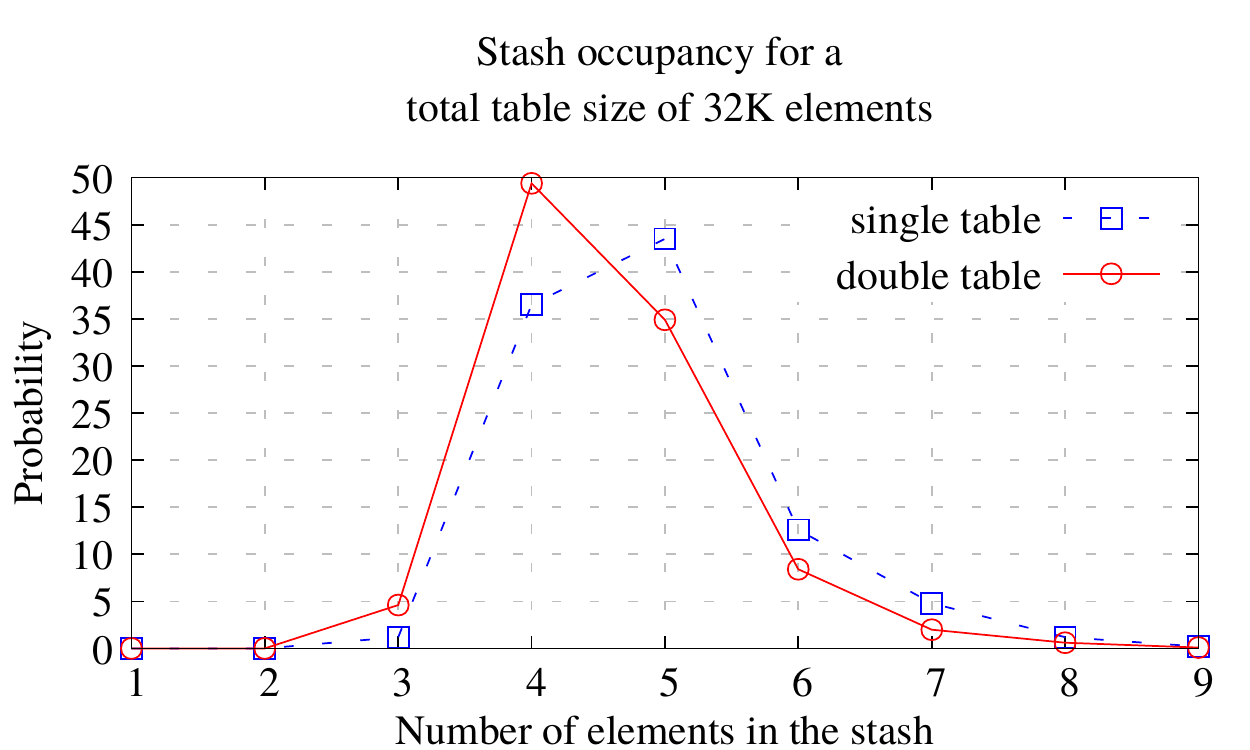} 
    \includegraphics[width=0.3\textwidth]{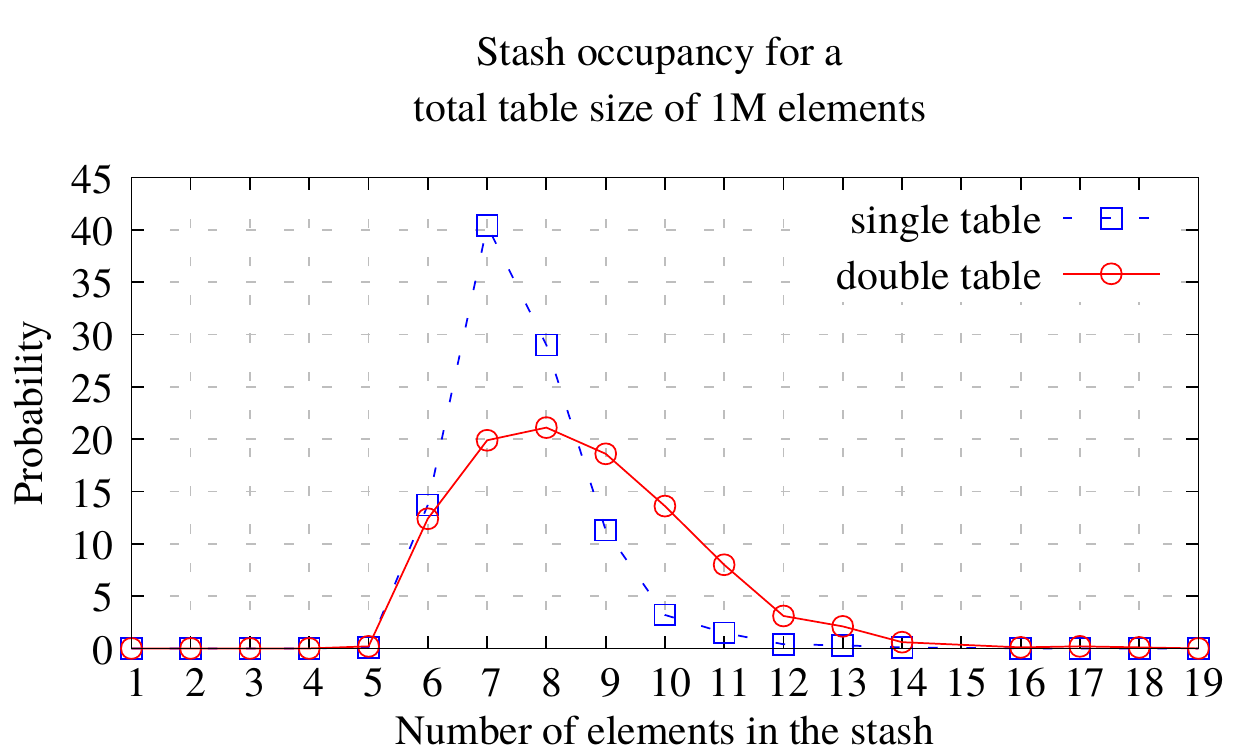} 
    \includegraphics[width=0.3\textwidth]{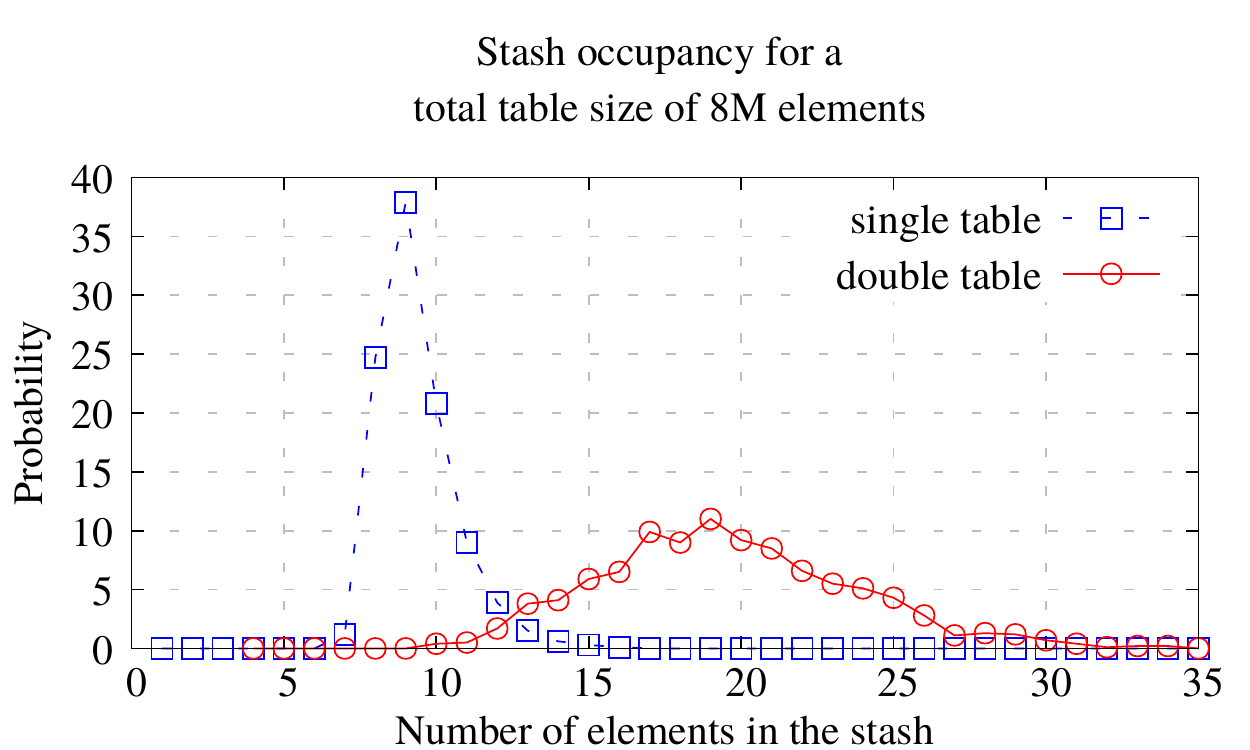} 
\caption{Probability distribution function for the maximum stash occupancy observed during the simulation at 95\% memory occupancy for $t=100$ and a total size of 32K, 1M, and 8M elements. 
}
\label{fig:graph3}
\end{figure*}


%



\subsection{Dynamic behavior at maximum load}

We conducted additional experiments for tables of size 8M to test performance with the insertion and removal of elements. We first load the hash table to 95\% memory occupancy, and then perform 16M replacement operations. The replacement first randomly selects an element in the EMOMA structure and removes it. Then it randomly creates a new entry (not already or previously present in the EMOMA) and inserts it.
This is a standard test for structures that handle insertions and deletions. The experiments were repeated 10 times, for both the single-table and double-table implementations. These experiments allow us to investigate the stability of the size of the stash in dynamic settings, near the maximum load.  Ideally, the stash size would remain almost constant in such dynamic settings. In Figure \ref{fig:graph4} we report the maximum stash occupancy observed. 
Each data point gives the maximum stash occupancy observed over the 10 trials over the last 1M trials;  that is, when the $x$-axis is 6, the data point is the maximum stash occupancy over replacements 5M to 6M over the 10 trials.

\begin{figure}[ht]
	\centering
	\includegraphics[width=0.4\textwidth]{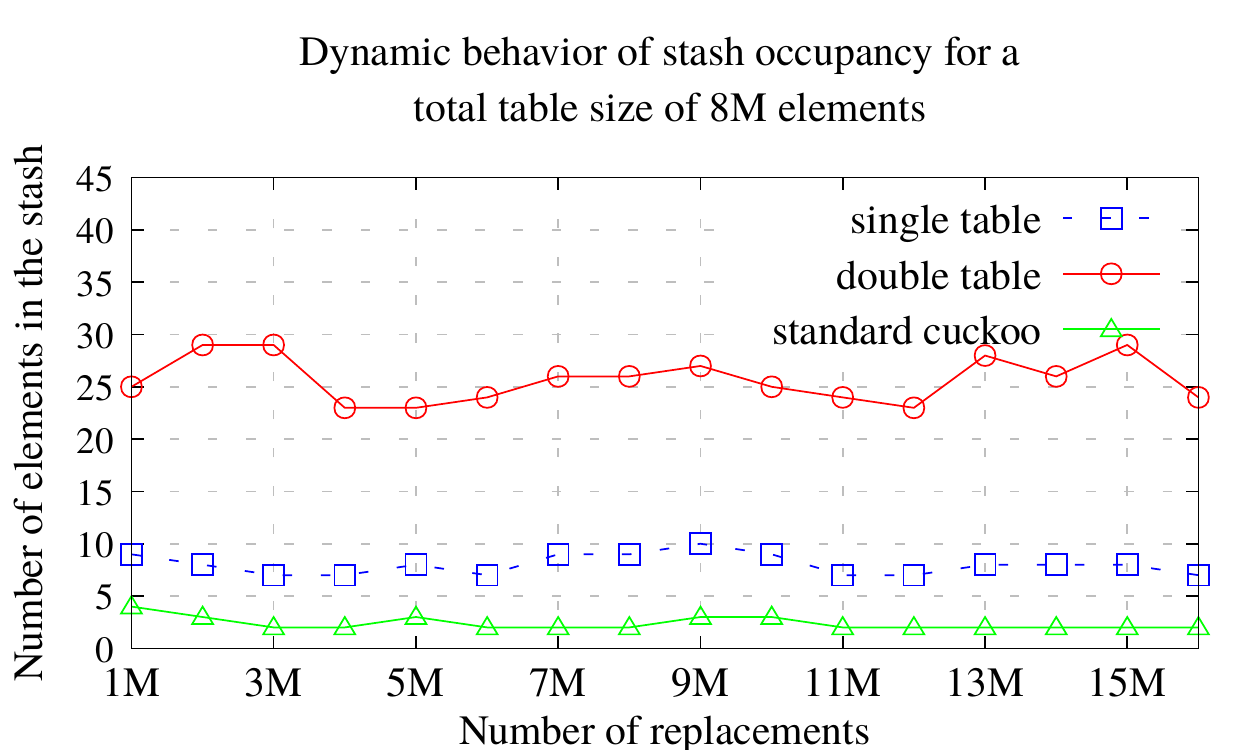}    
\caption{Maximum stash occupancy observed during insertion/removal for the standard cuckoo table, the single-table EMOMA and the double-table EMOMA implementations of total size of 8M elements with $t=100$.}

\label{fig:graph4}
\end{figure}

The experiments show that both implementations reliably maintain a stable stash size under repeated insertions and removals. The maximum stash occupancy observed over the 10 trials for the standard cuckoo table is in the range 1-4, for the single table EMOMA is always in the range 7-10, and for the double-table EMOMA setting it is in the range 23-29. This again shows that the single-table implementation provides better performance than the double-table, with a limited penalty in terms of stash size with respect to the standard cuckoo table.

\subsection{Insertion time}

The average number of iterations per insertion, which we also refer to as the average insertion time, can determine the frequency with which the EMOMA structure can be updated in practice, as the memory bandwidth needed to perform insertions is not available for query operations. The average insertion time depends both on $t$, the maximum number of iterations allowed for a single insertion, and on the load of the EMOMA structure. Larger $t$ allows for smaller stash sizes, as fewer elements are placed in the stash because they have run out of time when being inserted, but the corresponding drawback is an increase in the average insertion time. 
%

In Figure \ref{fig:graph5} we report the average number of iterations per insertion at different loads for $t=10,50,100,$ and 500 in tables of size 8M. 
The table is filled to the target load, and then 1M fresh elements are inserted by the same insertion/removal process described previously.  We measure the average number of iterations per insertion for the freshly inserted elements. The plots report the average insertion time for the single-table and double-table EMOMA configurations and for a standard cuckoo table.  


\begin{figure*}[t]
	\centering
	\includegraphics[width=0.32\textwidth]{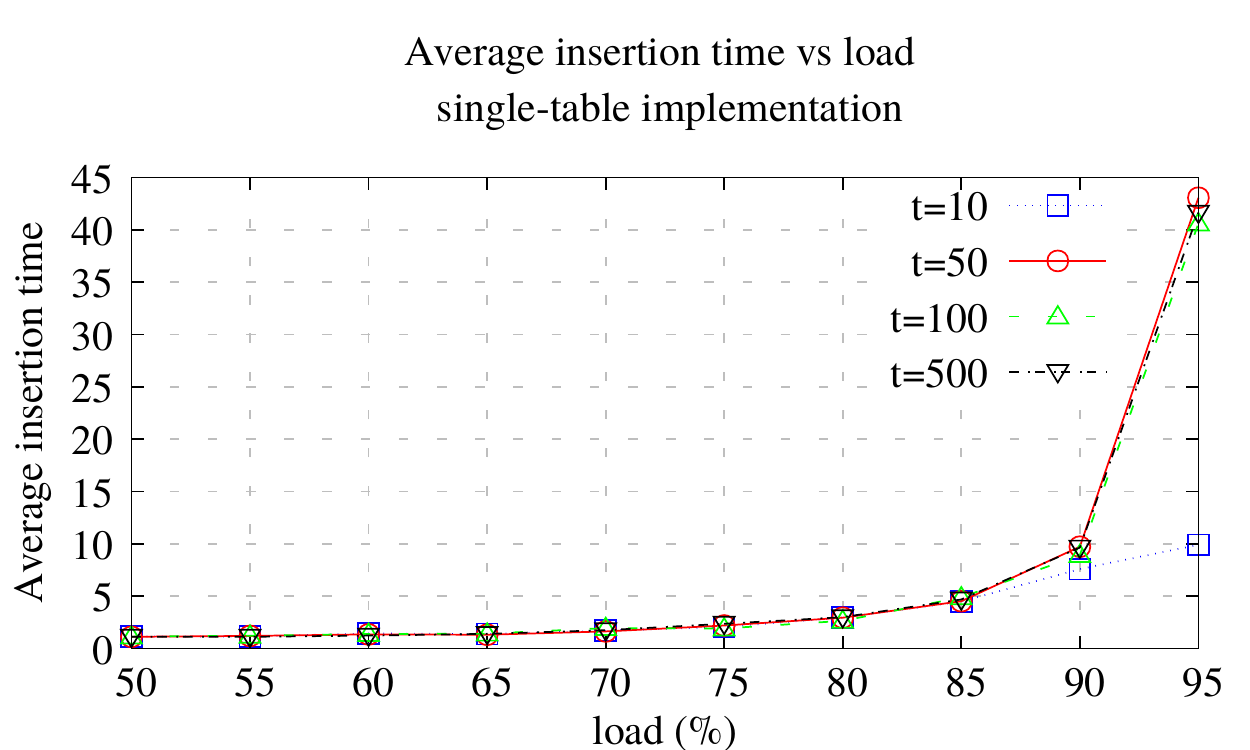}    
    \includegraphics[width=0.32\textwidth]{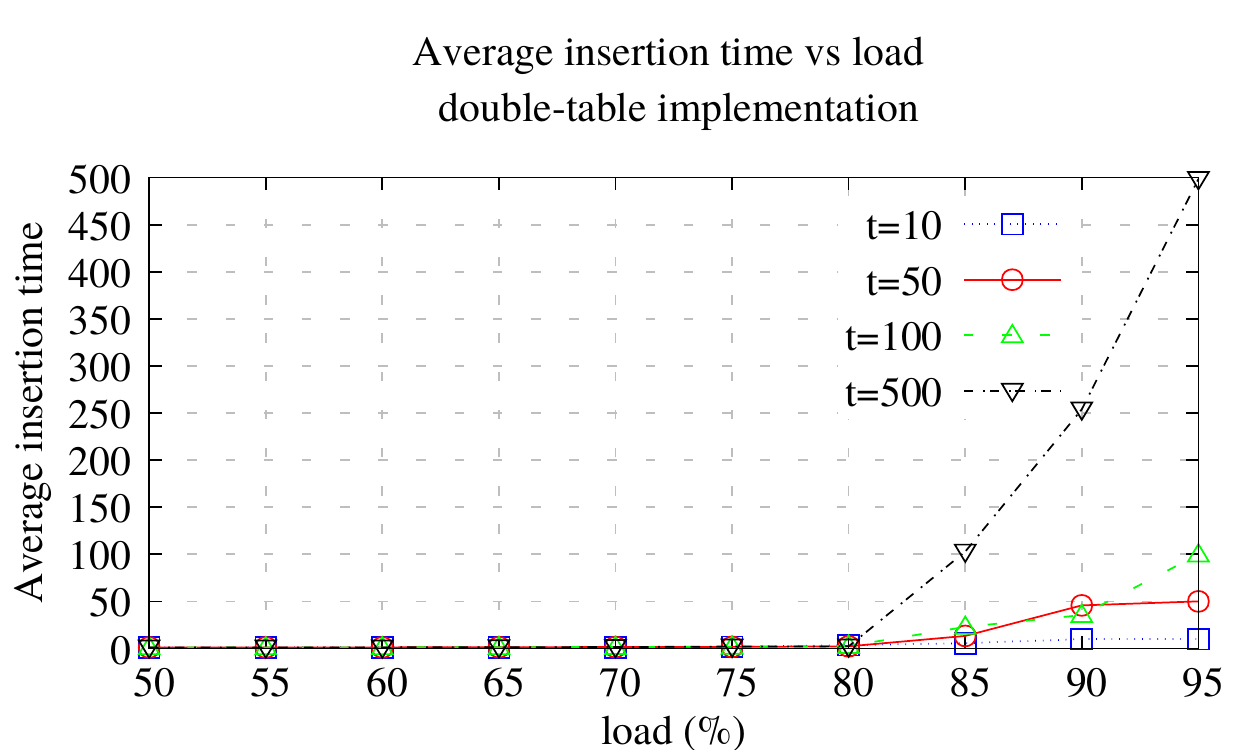}    
    \includegraphics[width=0.32\textwidth]{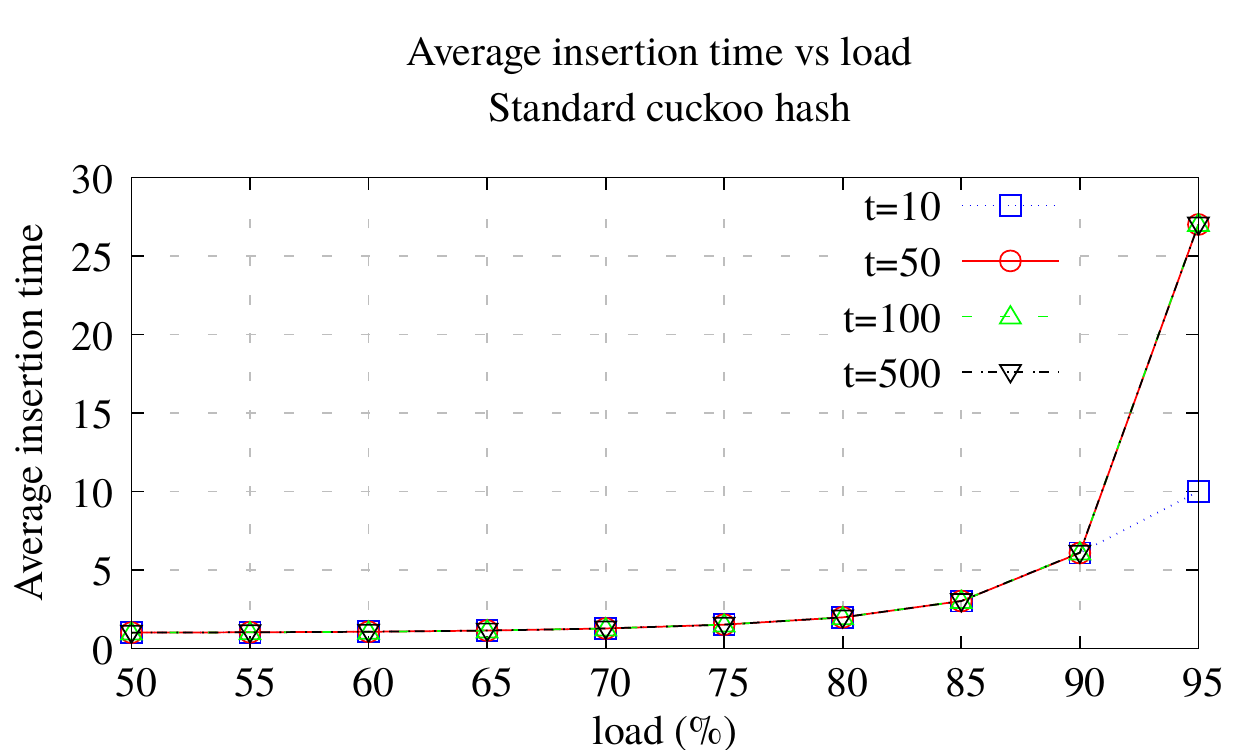}    
\caption{Average insertion time with respect to number of inserted elements (load) with different $t$ values. 
}
\label{fig:graph5}
\end{figure*}

As expected, the average insertion time increases substantially when the load increases to a point where the table is almost full. However, the behavior of the single-table and double-table configurations is significantly different (note the difference in the scale of the $y$-axes). For the single-table at maximum load (95\%) the average insertion time is almost equal to the maximum number of allowed iterations when $t=10$. This corresponds to a condition in which EMOMA is unable to complete insertions of new elements in $t$ steps, so elements remain in the stash, provoking an uncontrolled growth in the stash. With greater values of $t$, the system is able to insert the elements into the table in fewer than $t$ steps on average, with the average number of iterations per new element converging to around 44. In other words, in our tests when $t$ is at least 50, there will be some intervals where the stash empties, so the algorithm stops before reaching the maximum number of allowed iterations.  The single-table configuration can therefore work reliably when $t$ is set to values of at least 50. It is interesting to note that the results obtained for the single-table EMOMA configuration are qualitatively similar to those obtained for a standard cuckoo hash. In fact, for a standard cuckoo hash table the stash grows uncontrollably when $t=10$, but is stable when $t$ is at least 50. The average number of iterations per new element that is around 27 for the standard cuckoo hash table, so again we see the EMOMA implementation suffers a small penalty for the gain of knowing which of the two buckets an element lies in. Finally, it is interesting to note that the average number of iterations per new element also gives us an idea of the ratio of searches vs. insertions for which EMOMA is practical. For example, if the ratio is 1000 searches per insertion, then EMOMA requires only 4.4\% of the memory bandwidth for insertions.

For the double-table configuration, we instead see that the average insertion time remains almost equal to the  maximum number of allowed iterations. This means that the stash almost never empties, with some elements in the stash that the structure is either unable to place in the main table, or that stay in the stash for a large number of iterations. To avoid wasting memory accesses trying to place those elements, we could mark those elements and avoid attempts at moving them into the main table until a suitable number of replacements has been done. However, because we assume that the single-table implementation will be preferred due to its better performance, we do not explore this possibility further. 



\begin{figure}[ht]
	\centering
	\includegraphics[width=0.48\textwidth]{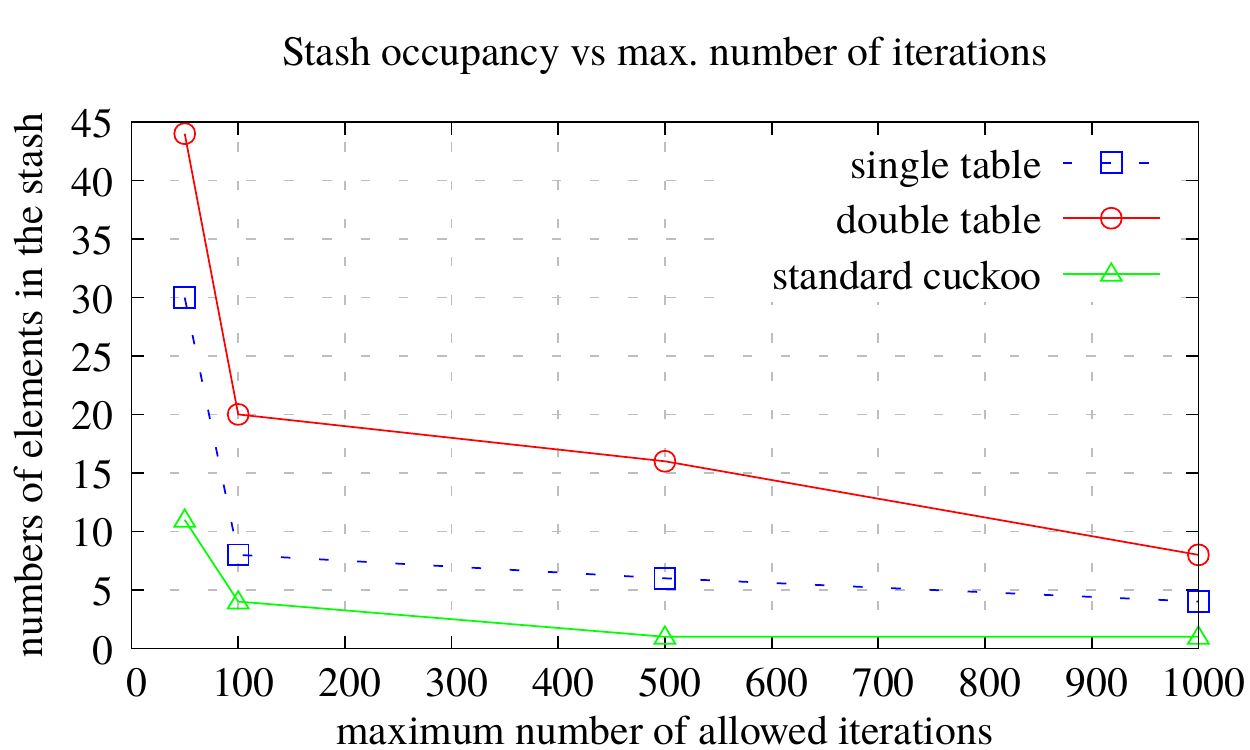}    
\caption{Maximum observed stash occupancy with respect to maximum number of allowed iterations $t$.}
\label{fig:graph6}
\end{figure}

To better understand the relationship between the maximum number of allowed iterations and the stash behavior, in Figure \ref{fig:graph6} we report the maximum stash occupancy observed over 100 trials at maximum load, for $t=50,100,500,$ and 1000, and for a table size of 8M elements. The graph reports the average insertion time for the single-table and double-table EMOMA configurations and for a standard cuckoo table.  
As expected, higher values of $t$ allow a smaller stash. The graph also shows that, with the same value of $t$, the single-table configuration requires fewer elements in the stash than the double-table configuration. 
The comparison with the standard cuckoo table shows that the standard cuckoo table does not actually need a stash if the number of allowed iterations is sufficiently large (the maximum value of 1 is due to the pending item that is moved during the insertion process), while the stash remains necessary for the EMOMA structures.  This is consistent with known results about cuckoo hashing \cite{21}. 


Summarizing, these experiments show that the single-table configuration provides better performance, but both configurations can work reliably even at the maximum target load of 95\%.


\subsection{Stash occupancy vs. table size}

The previous results suggest that a fairly small stash size is sufficient to enable a reliable operation of EMOMA when the single-table configuration is used. It is important to quantify how the maximum stash occupancy changes with respect to the table size in order to provision the stash to avoid overflow. We performed simulations to estimate the behavior of the failure probability with respect the table size and tried to extract some empirical rules. Obtaining more precise, provable numerical bounds remains an interesting theoretical open question.  Since we have already shown that the stash occupancy of the single-table configuration is significantly lower than that of the double-table configuration, we restricted the analysis only to the single-table case.

We performed 10,000 experiments where we fill the EMOMA table up to 95\% load and logged the maximum number of elements stored in the stash during the insertion phase. The simulation has been performed for table sizes 32K, 64K, 128K, 256K, 512K, 1M, 2M, 4M, and 8M. 

\begin{figure}[t]
	\centering
	\includegraphics[width=0.48\textwidth]{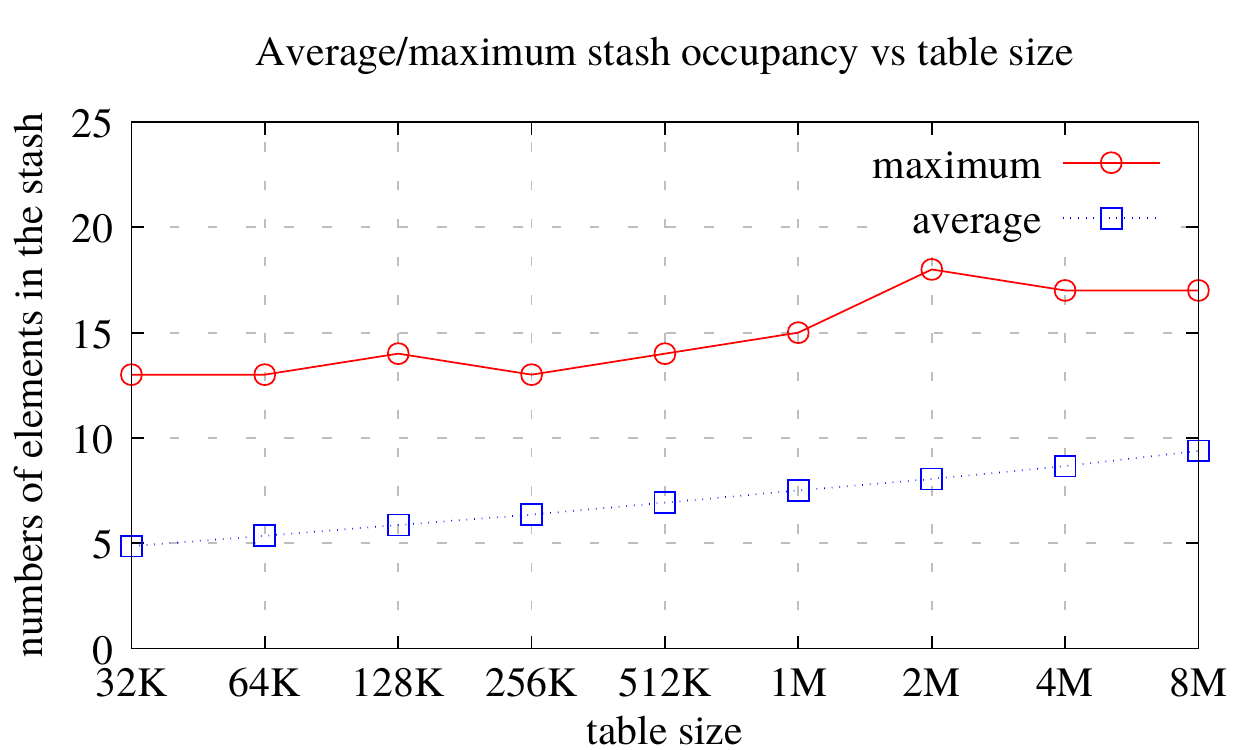}    
\caption{Average and maximum over 10,000 trials of the maximum number of elements in the stash with respect to table size for the single-table configuration}
\label{fig:graph7}
\end{figure}

Fig.~\ref{fig:graph7} presents the average maximum number of elements in the stash with respect to table size at the end of the insertion phase and the overall maximum stash occupancy observed over the 10,000 trials. As a rule of thumb, we can estimate that the average number of elements in the stash increases by 0.5 when the table size doubles. A similar trend occurs also for the  maximum stash occupancy observed over the 10,000 trials although in this case the variability is larger than for the average.

\begin{figure*}[t]
	\centering
	\includegraphics[width=0.48\textwidth]{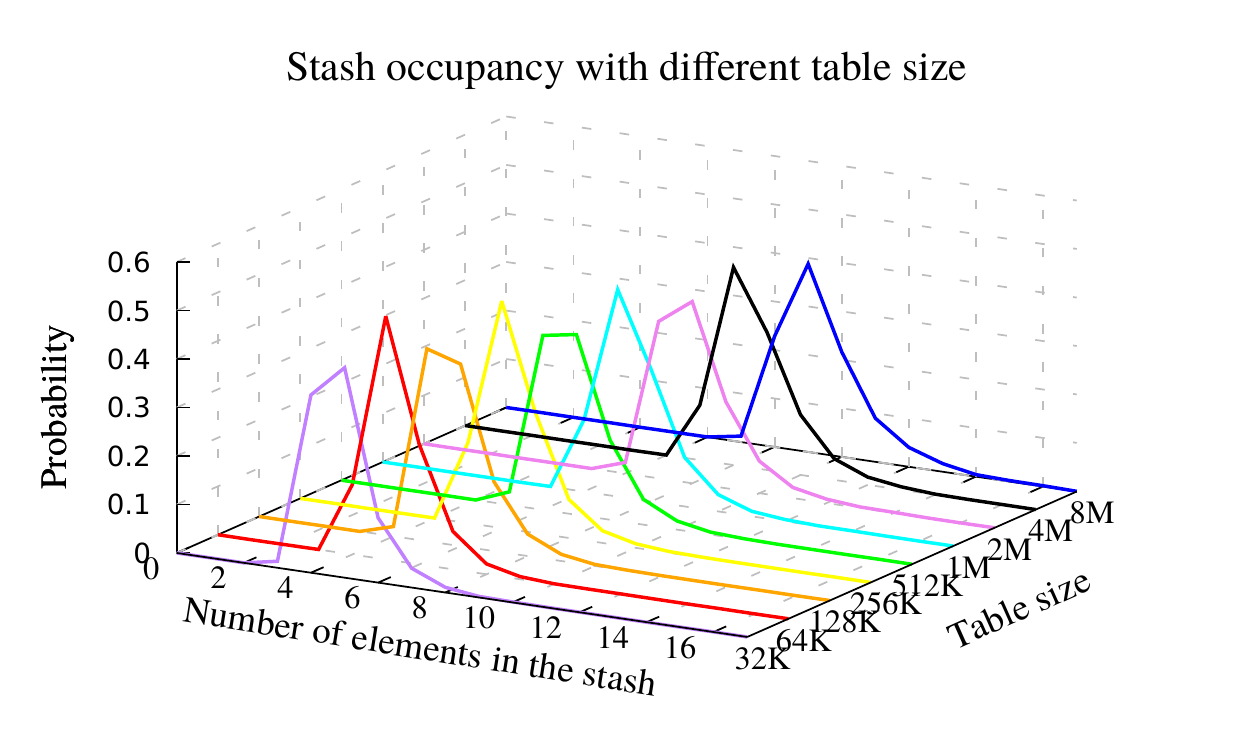}    
  	\includegraphics[width=0.48\textwidth]{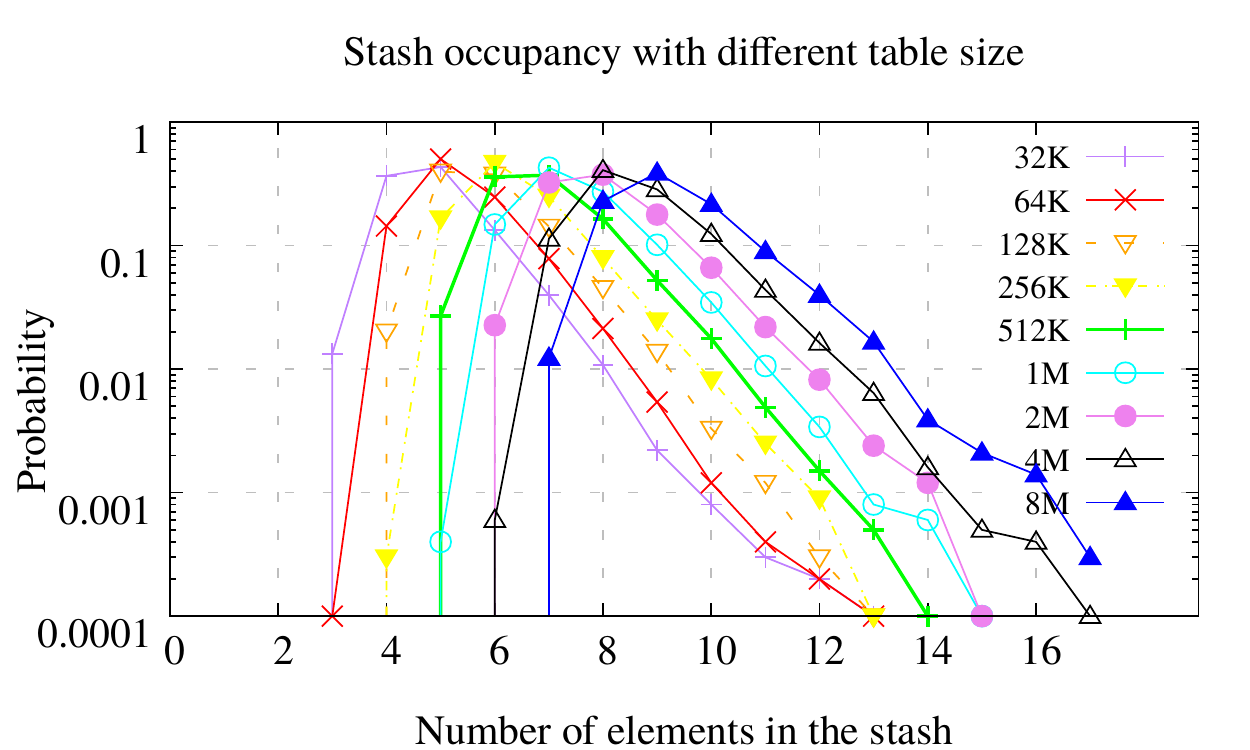}    
\caption{Probability distribution function for the maximum stash occupancy observed during the simulation at 95\% memory occupancy for different table size}
\label{fig:graph8}
\end{figure*}

Fig. \ref{fig:graph8} shows in linear and logarithmic scale the probability distribution function for the maximum stash occupancy for different table sizes over the 10,000 trials. As can be seen, after reaching a maximum value, the probability distribution function decreases exponentially with a slope that is slightly dependent on the table size.  A conservative estimate based on the empirical results is that beyond the average value for the maximum stash size, the probability of reaching a certain stash size falls by a factor of 10 as the stash size increases by 3 elements.

As an example of how to use this rule of thumb, we see that the empirically observed probability of having 17 or more elements in the stash for a table of size 8M at 95\% load is less than $10^{-3}$.  If a stash of size 16 fails with probability at most $10^{-3}$, by our rule of thumb we estimate a stash of size 31 would fail with probability at most $10^{-8}$, and a stash of size 64 would fail with probability at most $10^{-19}$.  While these are just estimates, they suggest that a stash that holds 64 elements will be sufficient for most practical scenarios.  

\section{Comparison with Alternative Approaches}

Most of the existing hash based techniques to implement exact match
have a worst case of more than one external memory access to complete
a lookup.  Such a worst case would hold for example for a hash table with separate
chaining or a standard cuckoo hash table.

The number of external memory accesses can be reduced by using an
on-chip approximate membership data structure that selects the
external positions that need to be accessed. In many cases
this does not result in a worst case of one memory access per lookup
due to false positives.  For example if a Bloom filter is used to
determine if a given position needs to be checked, a false positive
will cause an access to that position, even if the element is stored in another
position.  Other approaches to this problem have been proposed, namely the Fast Hash Table
(FHT) \cite{16} and the Bloomier filter \cite{28,29}.

In the Fast Hash Table with extended Bloom filter \cite{16}, $k$ hash
functions are used to map elements to $k$ possible positions in an
external memory. The same hash functions are used in an on-chip
counting Bloom filter.  Elements are then stored in the position (out
of the $k$) that has the smallest count value in the counting Bloom
filter.  If there are more than one position with the same count
value, the position with the smallest index is selected.  Then on a
search, the counting Bloom filter is checked and only that position is
accessed.  In most cases this method requires a single external memory
access, even under the assumption that a bucket holds only one element.  
(We assume an external memory access corresponds to a bucket of four
elements in our work above.)
However, when two (or more) elements are stored on the same
position (because it has the minimum count value for both), more than
one access may be required.  

The probability of
this occurring can be reduced by artificially increasing the counter
in those cases so that elements are forced to map to other positions.
In \cite{16}, the counting Bloom filter was dimensioned to have a size
$m$ that is 12.8 times the number of elements $n$ to be stored.  As
three bits were used for the counters this means that approximately 38
bits of on-chip memory are needed per element.  This is almost an
order of magnitude more than the on-chip memory required for EMOMA.
This difference arises because the counters have to be stored on-chip
and the load $n/m$ of the counting Bloom filter has to be well below
one for the scheme to work.  While this memory could be reduced for larger
bucket sizes, the on-chip memory use is still significantly larger
than ours in natural configurations.  Similarly, the off-chip memory is
significantly larger;  most buckets in the FHT schemes are necessarily empty.  
Finally, insertions and deletions are significantly more complex.  
Overall, the FHT approach takes more space and, being more complex, is much less amenable to a hardware implementation.  

Another alternative would be to use the approach we use in this paper,
but use a Bloomier filter \cite{28,29} in place of a counting
block Bloom filter to determine the position in external memory that
needs to be accessed.  A Bloomier filter is a data structure designed
to provide values for elements in a set; it can be seen as an
extension of a Bloom filter that provides not just membership
information, but a return value.  In particular, the output for a Bloomier
filter could be from $\{0,1\}$, denoting which hash function to use for
an element.  If a query is made for an element
not in the set, an arbitrary value can be returned; this feature of a
Bloomier filter is similar to the false positive of a Bloom filter.  
Moreover, a mutable Bloomier filter can be modified, so if an
element's position in the cuckoo table changes (that is the hash function used
for that element changes), the Bloomier filter can be updated in constant average time and logarithmic
(in $n$) time with high probability.  As a Bloomier filter provides
the exact response for elements in the set, only one external memory
access is needed; for elements not present in the set, at most one
memory access is also required, and the element will not be found.  Advantages of the Bloomier filter 
are that it allows the full flexibility of the choices in the cuckoo
hash table, so slightly higher hash table loads can be achieved.  It
can potentially also use less on-chip memory per element (at the risk of
increasing the probability needed for reconstruction, discussed below).   

However, the Bloomier filter comes with significant drawbacks.  First,
a significant amount ($\Omega(n \log n)$ under known constructions) of
additional off-chip memory would be required to allow a Bloomier
filter to be mutable.  Bloomier filters have non-trivial failure
probabilities; even offline, their failure probability is constant
when using space linear in $n$.  Hence, particularly under insertion and deletion
of elements, there is a small but non-trivial chance the Bloomier filter
will have to be reconstructed with new hash functions.  Such
reconstructions pose a problem for network devices that require high
availability.  Finally, the construction and update procedures of
Bloomier filters are more complex and difficult to implement in
hardware than our construction.  In particular, they require solving
sets of linear equations (that can be solved by back substitution) 
to determine what values to store so that the proper value is returned
on an element query, compared to the more simple operations of our 
proposed counting block Bloom filter.


Because of these significant issues, we have not implemented
head-to-head comparisons between EMOMA and these alternatives.  While
all of these solutions represent potentially useful data structures for
some problem settings,
for solutions requiring hardware-amenable designs using a single
off-chip memory access, EMOMA appears significantly better than these
alternatives.

\section{Hardware feasibility}

We have evaluated the feasibility of a hardware implementation of EMOMA using the NetFPGA SUME board \cite{9} as the target platform. The SUME NetFPGA is a well-known solution for rapid prototyping of 10Gb/s and
40Gb/s applications. It is based upon a  Xilinx Virtex-7 690T FPGA device and has four 10Gb/s Ethernet interfaces, three 36-bit QDRII+ SRAM memory devices running at 500MHz, and a DRAM memory composed of two 64-bit DDR3 memory modules running at 933MHz. We leverage the reference design available for the SUME NetFPGA to implement our scheme. In particular, the reference design contains a MicroBlaze (the Xilinx 32-bit soft-core RISC microprocessor) that is used to control the blocks implemented in the FPGA the using the AXI-Lite\cite{15} bus. The microprocessor can be used to perform the insertion procedures of the EMOMA scheme, writing the necessary values in the CBBF, in the stash, and in the external memories. 

Our goal here is to determine the hardware resources that would be used by an EMOMA scheme.  We select a key of 64 bits and an associated value of 64 bits. Therefore, each bucket of four cells has 512 bits. A bucket can be read in one memory access as a DRAM burst access provides precisely 512 bits. The main table has 524288 (512K) buckets of 512 bits requiring in total 256Mb of memory.  The stash is realized implementing on the FPGA a 64x64 bits Content Addressable Memory with the write port connected to the AXI bus and the read port used to perform the query operations. For the CBBF we used $k=4$ hash functions and a memory size of 524288 (512K) words of 16 bits. The memory of the CBBF uses two ports: the write port is connected to the AXI bus and the read port is used for the query operations. The results are reported in Table \ref{t:synth}. The table reports for each hardware block the number of LUTs (Look-Up Tables), the number of Flip-Flops, and the number of BRAMs (Block RAMs) used. We also show in parenthesis the percentage of resources used with respect to those available in the FPGA hosted in the NetFPGA board. For completeness, we report also the overhead of the MicroBlaze, even if it is not related only to the EMOMA scheme, as it is needed for almost any application built on top of the NetFPGA. It can be observed that EMOMA needs only a small fraction of the FPGA resources. As expected, the most demanding block is the memory for the CBBF, which in this case requires 256  (17\%) of the 1470 available Block RAMs.

\begin{table}[t]
\centering
\caption{Hardware cost of EMOMA components}
\begin{tabular}{|c|c|c|c|}
  \hline
EMOMA component &  \#LUTs  & Flip-Flops & \#BRAM  \\
\hline
Stash & 3337 (1.12\%) & 102 ($<$  0.01\%) & 1 ($<$ 0.01\%) \\
 \hline
CBBF &  61 ($<$ 0.01\%) & 1  ($<$ 0.01\%) & 256 (17.41\%) \\  
  \hline
MicroBlaze &  882 (0.27\%) & 771  (0.09\%) & 32 (2.18\%) \\  
  \hline
\end{tabular}
\label{t:synth}
\end{table}

Finally, the insertion procedure has been compiled for the MicroBlaze architecture and the code footprint is around 30KB of code. This is a fairly small amount of memory, since the instruction memory size of the MicroBlaze can be configured to be larger than 256KB. As a summary, this initial evaluation shows that EMOMA can be implemented on an FPGA based system with limited cost.


\section{Conclusions and Future Work}

We have presented Exact Match in One Memory Access (EMOMA), a scheme
that implements exact match with only one access to external memory,
targeted towards hardware implementations of high availability network
processing devices such as switches or routers.  EMOMA uses a counting
block Bloom filter to select the position that needs to be accessed in
an external memory cuckoo hash table to find an element.  By sharing
one hash function between the cuckoo hash table and the counting block
Bloom filter, we enable fast identification of the elements that can
create false positives, allowing those elements to be moved in the
hash table to avoid the false positives. This requires a 
few additional memory accesses for some insertion operations and a slightly 
more complex insertion procedure. Our evaluation shows that
EMOMA can achieve around 95\% utilization of the external memory
when using only slightly more than 4 bits of on-chip memory for each
element stored in the table.  This compares quite favorably with previous
schemes such as Fast Hash Table \cite{16}, and is also simpler for 
implementation.

A theoretical analysis of EMOMA remains open, and might provide
additional insights on optimization of EMOMA. Another idea to explore
would be to generalize EMOMA so that instead of the same hash function
being used for the counting block Bloom filter and the first position
in the cuckoo hash table, only the higher order bits of that function
were used for the CBBF.  This would mean several buckets in the cuckoo
hash table would map to the same block in the CBBF, providing
additional trade-offs.  In particular, this would lead to EMOMA
configurations using fewer bits of on-chip memory per element, but
would lead to more complex insertion operations. 

\section{Acknowledgments}
Salvatore Pontarelli is partially supported by the European Commission in the frame of the BEBA project \\
http://www.beba-project.eu/.
Pedro Reviriego would like to acknowledge the support of the excellence network Elastic Networks TEC2015-71932-REDT funded by the Spanish Ministry of Economy and Competitivity.  Michael Mitzenmacher was supported in part by NSF grants CNS-1228598, CCF-1320231, CCF-1535795, and CCF-1563710.  

\bibliographystyle{IEEEtran}

\begin{thebibliography}{99}

\bibitem{1}	P. Gupta and N. McKeown, ``Algorithms for packet classification,'' IEEE Network, vol. 15 no. 2, pp. 24-32, 2001.
\bibitem{2}	K. Pagiamtzis and A. Sheikholeslami, ``Content-addressable memory (CAM) circuits and architectures: a tutorial and survey,'' IEEE Journal of Solid-State Circuits, vol. 41, no. 2, pp. 712-727, 2006.
\bibitem{3}	F. Yu, R.H. Katz and T.V. Lakshman, ``Efficient multimatch packet classification and lookup with TCAM,'' IEEE Micro, vol. 25, no. 1, pp. 50-59, 2005.
\bibitem{4}	A. Kirsch, M. Mitzenmacher and G. Varghese, ``Hash-based techniques for high-speed packet processing,'' Algorithms for Next Generation Networks,  pp. 181-218, Springer London, 2010. 
\bibitem{5}	R. Pagh and F. F. Rodler, ``Cuckoo hashing,'' Journal of Algorithms, pp. 122-144, 2004.
\bibitem{6}	M. Waldvogel, et al. ``Scalable high speed IP routing lookups,'' in Proc. of the Conference on Applications, Technologies, Architectures, and Protocols for Computer Communications (SIGCOMM), pp. 25-36,  1997.
\bibitem{7}	W. Jiang, Q. Wang, and V. Prasanna ``Beyond TCAMs: An SRAM based Parallel Multi-Pipeline Architecture for Terabit IP Lookup,'' in Proc of the 27th Conference on Computer Communications (INFOCOM), pp. 1786-194, 2008.
\bibitem{8}	P. Bosshart, G. Gibb, H. S, Kim, G. Varghese, N. McKeown, M. Izzard, F. Mujica, and M. Horowitz, ``Forwarding metamorphosis: fast programmable match-action processing in hardware for SDN,'' in Proc. of the Conference on Applications, Technologies, Architectures, and Protocols for Computer Communications (SIGCOMM), pp. 99-110, 2013.
\bibitem{9}	N. Zilberman, Y. Audzevich, G. Covington and A. Moore, ``NetFPGA SUME: Toward 100 Gbps as research commodity,'' IEEE Micro, vol. 34, pp. 32-41, 2014.
\bibitem{10}	Y. Kanizo, D. Hay and I. Keslassy, ``Maximizing the Throughput of Hash Tables in Network Devices with Combined SRAM/DRAM Memory,'' IEEE Transactions on Parallel and Distributed Systems, vol. 26, no. 3, pp. 796-809, 2015.
\bibitem{Binkert11} N. Binkert, A. Davis, N.P. Jouppi, M. McLaren, N. Muralimanohar, R. Schreiber, and J.H. Ahn ``The role of optics in future high radix switch design,' in Proc. of 38th IEEE International Symposium on Computer Architecture (ISCA), pp. 437-447, 2011.
\bibitem{11} S. Dharmapurikar, P. Krishnamurthy, and D. E. Taylor, ``Longest prefix matching using Bloom filters,'' IEEE/ACM Transactions on Networking, vol. 14, no. 4, pp 397-409, 2006.
\bibitem{CheapSilicon13} G. Pongr\'{a}cz,  L Moln\'{a}r and ZL Kis, Z Tur\'{a}nyi ``Cheap silicon: a myth or reality? picking the right data plane hardware for software defined networking,'' in Proc. of the 2nd ACM SIGCOMM Workshop on Hot Topics in Software Defined Networking, pp. 103-108, 2013.
\bibitem{12} B. Sinharoy, et al. ``IBM POWER8 processor core microarchitecture'', IBM Journal of Research and Development, vol. 59, no. 1, pp. 2,1-21,  2015.
\bibitem{samsung} Samsung 2Gb SDRAM data sheet, available online: http://www.samsung.com/global/business/semiconductor/ file/2011/product/2011/8/29/729200ds\_k4b2gxx46d\_rev113.pdf
\bibitem{Iyer03} S. Iyer and N. McKeown. ``Analysis of the parallel packet switch architecture,'' IEEE/ACM Transactions on Networking, vol, 11, no. 2, pp. 314-324, April 2003.

\bibitem{13} Micron RLDRAM 3 data sheet, available online: https://www.micron.com/~/media/documents/ products/data-sheet/dram/576mb\_rldram3.pdf
\bibitem{14} Cypress QDR-IV SRAM data sheet, available online:  http://www.cypress.com/documentation/datasheets/ cy7c4022kv13cy7c4042kv13-72-mbit-qdr-iv-xp-sram
\bibitem{16} H. Song, S. Dharmapurikar, J. Turner and J. Lockwood, ``Fast hash table lookup using extended Bloom filter: an aid to network processing,'' ACM SIGCOMM Computer Communication Review, vol. 35, no. 4, pp. 181-192, 2005.
\bibitem{17} S. Pontarelli, P. Reviriego and J.A. Maestro, ``Parallel d-Pipeline: a Cuckoo Hashing Implementation for Increased Throughput,'' IEEE Transactions on Computers, vol. 65, no 1, pp. 326-331, Jan. 2016.
\bibitem{18}
M. Dietzfelbinger, A. Goerdt, M. Mitzenmacher, A. Montanari, R. Pagh, and M. Rink,
``Tight thresholds for cuckoo hashing via XORSAT,'' in Proc. of ICALP, pp. 213-225, 2010.
\bibitem{19}
J. Cain, P. Sanders, and N. Wormald,
``The random graph threshold for k-orientiability and a fast algorithm for optimal multiple-choice allocation,''
in Proc. of the Eighteenth Annual ACM-SIAM Symposium on Discrete Algorithms, pp. 469-476, 2007.
\bibitem{20}
D. Fernholz and V. Ramachandran,
``The k-orientability thresholds for $G_{n,p}$,''
in Proc. of the Eighteenth Annual ACM-SIAM Symposium on Discrete Algorithms, pp. 459-468, 2007.
\bibitem{21}  A. Kirsch, M. Mitzenmacher, and U Wieder,
``More robust hashing: Cuckoo hashing with a stash,'' 
SIAM Journal on Computing, vol. 39, no. 4, pp. 1543-1561, 2009.
\bibitem{22} A. Kirsch, and M. Mitzenmacher,
``Using a queue to de-amortize cuckoo hashing in hardware,''
in Proc. of the Forty-Fifth Annual Allerton Conference on Communication, Control, and Computing, 2007.
\bibitem{23} B. Bloom, ``Space/time tradeoffs in hash coding with allowable errors,'' Communications of the ACM, vol. 13, no. 7, pp. 422-426, 1970.
\bibitem{24} A. Broder and M. Mitzenmacher, ``Network applications of Bloom filters: A survey,'' Internet Math., vol. 1, no. 4, pp. 485-509, 2003.
\bibitem{25} U. Manber and S. Wu, ``An algorithm for approximate membership checking with application to password security,'' Information Processing Letters, vol. 50, no. 4, pp. 191-197, 1994.
\bibitem{bgp1} G. Huston, and A. Grenville ``Projecting future IPv4 router requirements from trends in dynamic BGP behaviour,'', in Proc. of the Australian Telecommunication Networks and Applications Conference (ATNAC), 2006.
\bibitem{bgp2} A. Elmokashfi, A. Kvalbein and C.Dovrolis, ``On the scalability of BGP: the roles of topology growth and update rate-limiting,'' in Proc. of the ACM CoNEXT Conference, 2008.
\bibitem{28} B. Chazelle, J. Kilian, R. Rubinfeld, and A. Tal, ``The Bloomier filter: an efficient data structure for static support lookup tables,'' in Proc. Fifteenth Annual ACM-SIAM Symposium on Discrete Algorithms, pp. 30-39, 2004. 
\bibitem{29} D. Charles and K. Chellapilla, ``Bloomier filters: A second look,'' in Proc. 16th Annual European Symposium on Algorithms, pp. 259-270, 2008.
\bibitem{15} AXI Reference Guide - Xilinx
https://www.xilinx.com/support/documentation/ ip\_documentation/ug761\_axi\_reference\_guide.pdf






\end{thebibliography}

\end{document}